\documentclass{fundam}

\publyear{25}
\papernumber{11}
\volume{195}
\issue{1-4}
\theDOI{10.46298/fi.12436}

\versionForARXIV

\usepackage[utf8]{inputenc}
\usepackage{amsmath,amssymb,mathtools}
\usepackage{xcolor}
\usepackage{booktabs}
\usepackage{textcomp}
\usepackage{pgf}
\usepackage{upgreek}

\newcommand\Dim{\mathrm{Dm}}
\newcommand\Unit{\mathrm{Ut}}
\newcommand\Pref{\mathrm{Px}}

\newcommand\subbase{_{\mathrm{b}}}
\newcommand\subroot{_{\mathrm{r}}}
\newcommand\subpre{_{\mathrm{p}}}
\newcommand\subnorm{_{\mathrm{n}}}
\newcommand\subeval{_{\mathrm{e}}}
\newcommand\subabs{_{\mathrm{a}}}

\newcommand\subzero{\!_{/0}}

\newcommand\RatM{\mathrm{Qm}}
\newcommand\RatP{\mathbb{Q}_+}
\newcommand\FinSupp{\mathrm{Zf}}
\newcommand\Forget{\mathcal{U}}
\newcommand\Abel{\mathcal{A}}
\newcommand\Pair[1][]{\mathcal{C}_{#1}}
\newcommand\GPair[1][]{\mathcal{D}_{#1}}

\newcommand\conv[4][]{#2 \xrightarrow{#3}_{#1} #4}

\newcommand\de[1]{\lfloor #1\rfloor}

\newcommand\Usym[1]{\mathrm{#1}}

\renewcommand\_{\underline{\enspace}}

\usepackage[all]{xy}

\newcommand\NEW[1]{#1}
\newenvironment{ENEW}{}{}

\newcommand\NOW[1]{#1}
\newenvironment{ENOW}{}{}

\title{A Theory of Conversion Relations for Prefixed Units of Measure}
\author{Baltasar Trancón y Widemann \\
  Technische Hochschule Brandenburg \\
  Brandenburg an der Havel, DE \\
  \NEW{trancon{@}th-brandenburg.de}
  \and
  Markus Lepper \\
  semantics gGmbH \\
  Berlin, DE}
\begin{document}
\maketitle

\runninghead{B. Trancón y Widemann, M. Lepper}{A Theory of Conversion Relations for Prefixed Units of Measure}

\begin{abstract}
  Units of measure with prefixes and conversion rules are given a formal
  semantic model in terms of categorial group theory.  Basic structures and both
  natural and contingent semantic operations are defined.  Conversion rules are
  represented as a class of ternary relations with both group-like and
  category-like properties.  A hierarchy of subclasses is explored, each
  \NEW{satisfying stronger useful algebraic properties} than the preceding,
  culminating in a direct efficient conversion-by-rewriting algorithm.
\end{abstract}


\section{Introduction}

In the mathematics of science, \emph{dimensions} and \emph{units of measure} are
used as static metadata for understanding and checking \emph{quantity} relations
\cite{Wallot1922}.  Quantities are variables that may assume a value expressing
a fact about a model, in terms of a formal multiplication of a \emph{numerical
  value} with a unit of measure.\footnote{~\NEW{In many fields of mathematics
    that deal with scalable entities, the former part would be called a
    \emph{magnitude}; however, in the context of quantities, the term
    \emph{magnitude} is applied to the pair.  Hence we avoid that term for the
    sake of clarity.}}  For example, $42\,195\;\mathrm{m}$ is a quantity value
of \emph{distance}, \NEW{with numerical value $42\,195$ and unit of measure
  $\mathrm{m}$}.  Quantities themselves have a complicated and controversial
ontology, which shall not be discussed here.

Quantities and quantity values are quite distinct epistemological types, and the
need to distinguish them properly is demonstrated beyond doubt by Partee's
Paradox \cite{Loebner}: ``The temperature is ninety'' and ``the temperature is
rising'' do not entail ``ninety is rising.''  The former premise is a statement
about (an instantaneous assignment of) a quantity value; note that the
\NEW{numerical value} is explicit but the unit is implicit.  The latter premise
is a statement about a quantity that persists over time.  The conclusion, being
in type error, is not a meaningful statement about anything physical.

Clearly, neither half of the \NEW{numerical value}--unit pair is sufficient for
interpretation: The \NEW{numerical value} carries the \NEW{relative} data, and
the unit of measure the context of reference.  \NEW{A system of quantities
  assigns a dimension to each quantity}, such that only quantities of identical
dimension are \emph{commensurable}, that is, may be added or compared (as apples
with apples), and may furthermore be associated with preferred units of
measurement.  For example, the quantities of \emph{molecular bond length} and
\emph{planetary perihelion} are both of dimension \emph{length} and thus
formally commensurable, but traditionally use quite different units of measure.

In scientific programming, traditional practice expresses only the numerical value
part of quantity values, whereas units of measure, and their consistency among
related quantities, are left implicit, just like in the quotation above.  The
evident potential for disastrous software errors inherent in that practice has
been demonstrated, for instance, by the famous crash of the Mars Climate Orbiter
probe, where \emph{pounds of force} were confounded with \emph{newtons} at a
subsystem interface \cite{mco}.  Practical tool solutions abound
\cite{McKeever2020}.  Theoretical foundations are rare, but essential for
program correctness, specification and verification.

In the seminal work of Kennedy~\cite{Kennedy1996}, the checking of stated or
implied units of measure is cast formally as a \emph{type inference} problem.
The approach has been implemented successfully, for example in
F\#~\cite{Kennedy2010} or Haskell~\cite{Gundry2015}.  However, the underlying
concept of units is strongly simplified, and falls far short of the traditional
scientific practice, embodied in the International Systems of Units and
Quantities (SI, ISQ~\cite{ISO80000-1}, respectively).

A particular challenge is the concept of \emph{convertible} units, where the
scaling factors for equivalent transformation of \NEW{numerical values} relative
to some other unit are fixed and known, and hence could be applied implicitly.
In the Mars Climate Orbiter case, implicit multiplication by
$4.448\,221\,615\,260\,5$ would have avoided the observed failure, and thus
potentially saved the mission.

Furthermore, as a special case of convertibility of such practical importance
that it comes with its own notation, there are \emph{prefixes} that can be
applied to any unit, and modify the scaling factor in a uniform way.  For
example, multiplication by $1000$ serves to convert from seconds to
milliseconds, and exactly in the same way from meters to millimeters.

\subsection{A Survey of Problem Cases}

Traditional concepts of units of measure also have their own curious,
historically grown idiosyncrasies and restrictions.  The following five problem
cases shall serve to illustrate the need for more rigorous formalization, and
also as challenges for the explanatory power of the model to be developed in the
main matter of the present paper, to be recapitulated in
Section~\ref{cases-rev}.

\subsubsection{Problem Case: Neutral Elements}
\label{1-aargh}
  
\begin{ENOW}
  The treatment of neutral elements of the algebras of dimensions and units in
  the standard literature exhibits several subtle inexactitudes and
  inconsistencies.

  For the neutral element of the algebra of dimensions, the version of the ISQ
  that we are citing throughout this paper discusses the form
  \begin{quote}
    ``[\dots] $A^0B^0C^0\dots = 1$, where the symbol $1$ denotes the corresponding
    dimension.  Such a quantity is called a quantity of dimension one and is
    expressed by a number.''\;\cite[§5]{ISO80000-1}
  \end{quote}
  whereas the slightly later German version \cite{ISO80000-1-de} replaces
  ``dimension one'' by ``dimension number''.\;\footnote{~\NOW{``Eine solche Größe
      wird als \emph{Größe der Dimension Zahl} bezeichnet, und ihr Wert wird
      durch eine Zahl angegeben.''}\;\cite[§5]{ISO80000-1-de}} That this is not
  merely a liberty of translation, but a conscious choice is reflected by
  remarks to the effect that the terms ``quantity of dimension one'' and
  ``dimensionless quantity'' are deprecated\footnote{~\NOW{``Die Benennungen
      `Größe der Dimension Eins' und `dimensionslose Größe' sind veraltet und
      sollten nicht mehr verwendet werden.''}\;\cite[p.~3]{ISO80000-1-de}}, and
  that $1$ as the expression of a dimension is not allowed\footnote{~\NOW{``Die
      Zahl $1$ ist keine Dimension.  Hier sollte korrekterweise `Dimension Zahl'
      stehen.''}\;\cite[footnote N12]{ISO80000-1-de}\\\NOW{``Die Zahl $1$ ist
      keine Dimension.  Ein Zeichen für die Dimension Zahl ist derzeit noch
      nicht festgelegt.''}\;\cite[footnote N14]{ISO80000-1-de}}.

  A recent revision of \cite{ISO80000-1} sums up the controversy with the more
  cautious, but rather indecisive statement:
  \begin{quote}
    ``[\dots] $A^0B^0C^0\dots = 1$, where the symbol $1$ denotes the
    corresponding dimension.  There is no agreement on how to refer to such
    quantities.  They have been called dimensionless quantities (although this
    term should now be avoided), quantities with dimension one, quantities with
    dimension number, or quantities with the unit one.  Such quantities are
    dimensionally simply numbers.  To avoid confusion, it is helpful to use
    explicit units with these quantities where possible, e.g.,
    $\mathrm{m}/\mathrm{m}$ [\dots]''\;\cite[§5]{ISO80000-1-new}
  \end{quote}
  Note the self-contradiction between ``quantities with the unit one'' and ``use
  explicit units''.
  
  With regard to the neutral element of the algebra of units, the SI has a
  precise notion of its semantic status,
  \begin{quote}
    ``The unit one is the neutral element of any system of units---necessary and
    present automatically.''\;\cite[§2.3.3]{sib9}
  \end{quote}
  but still, somewhat contradictory, classifies it as derived by means of a
  syntactic criterion:
  \begin{quote}
    ``All other units [i.e.\ other than the seven base units], described as
    \emph{derived units}, are constructed as products of powers of the base
    units.''\;\cite[§1.2]{sib9}.
  \end{quote}
  
  The ISQ likewise classifies the neutral unit as not a base unit, but
  acknowledges that there are non-derived quantities measured in that unit,
  namely \emph{counts} of entities.  This is possible by virtue of the laxer,
  non-exclusive characterization of derived quantities (note the omission of
  ``all''):
  \begin{quote}
    ``Other quantities, called derived quantities, are defined or expressed in
    terms of base quantities by means of equations.''\,\cite[§4.3]{ISO80000-1}
  \end{quote}

  The notation rules common to SI and ISQ state that the neutral unit has the
  symbol $1$, but should not be written as such in a quantity value expression.
  It follows that, unlike any other unit symbol with a definition, it cannot
  take a prefix.
\end{ENOW}

\subsubsection{Problem Case: Equational Reasoning}
\label{eq-aargh}

The ISQ standard defines derived units by equations, equating each to some
expression over other, more basic units \cite[§6.5]{ISO80000-1}.  However,
equals often cannot be substituted for equals, without making arbitrary choices,
or violating semantic conventions or even basic syntax:
\begin{itemize}
\item Applying the prefix \emph{micro-} ($\upmu$) to the SI ``base'' unit
  \emph{kilogram} ($\mathrm{kg}$) yields the syntactically illegal
  $\upmu\mathrm{kg}$; the recommended, semantically equivalent alternative
  is the \emph{milligram} ($\mathrm{mg}$).
\item A unit of \emph{volume}, the \emph{centilitre}, is defined as the
  prefix \emph{centi-}, implying a multiplicator of $0.01$, being applied
  to the unit \emph{litre} ($\mathrm{L}$), which in turn defined as the
  cubic \emph{decimeter} ($\mathrm{dm}^3$).  The expanded expression
  $\mathrm{cL} = \mathrm{c}(\mathrm{dm}^3)$ is syntactically invalid, and there is
  no semantically equivalent expression in SI syntax proper; the equation
  $10^{-5} = 10^{3k}$ obviously has no integer solutions.
\item The coherent SI unit of \emph{power}, the \emph{watt} ($\mathrm{W}$)
  is defined, formally via several reduction steps involving other units,
  as $\mathrm{kg}{\cdot}\mathrm{m}^2{\cdot}\mathrm{s}^{-3}$.  For the
  \emph{hectowatt}, there are no less than $44$ distinct semantically
  equivalent alternatives of prefix distribution, such as
  $\mathrm{fg}{\cdot}\mathrm{cm}^2{\cdot}\mathrm{ns}^{-3}$, or
  $\mathrm{dg}{\cdot}\mathrm{\upmu m}^2{\cdot}\mathrm{\upmu s}^{-3}$, or
  $\mathrm{Eg}{\cdot}\mathrm{cm}^2{\cdot}\mathrm{ks}^{-3}$, neither of
  which can be called obviously preferable.
\end{itemize}
      
Besides the apparently arbitrary restrictions of syntactic expressibility, there
is the nontrivial problem of getting the circumventing semantic equivalence
right; see Section~\ref{unit-dim-aargh} below.

\subsubsection{Problem Case: Cancellation}
  
The ISQ standard defines the units of \emph{plane angle}, radians, and of
\emph{solid angle}, steradians, as
\begin{align*}
  \Usym{rad} &= \Usym{m}/\Usym{m} &
  \Usym{sr} &= \Usym{m}^2/\Usym{m}^2
\end{align*}
respectively, but considers them different from each other, and from the unit
$1$ \cite[§6.5]{ISO80000-1}.  Thus, the operation written as unit division
appears cancellative in many, but not nearly all contexts.

\subsubsection{Problem Case: Prefix Families}
\label{prefix-aargh}
  
The SI prefix families are geometric sequences, but cannot be written as powers
of a generator.  For example, the General Conference on Weights and Measures
(CGPM) has recently adopted the metric prefixes \emph{ronna-} and \emph{quetta-}
with the values $10^{27}$ and $10^{30}$, respectively, as well as their inverses
\emph{ronto-} and \emph{quecto-} \cite{CGPM2022R3}.  These are clearly intended
to represent the ninth and tenth powers, respectively, of the generating prefix
\emph{kilo-}, but cannot be written as such:
\begin{ENEW}
  Defining equations of the form
  \begin{align*}
    \Usym{R} &= \Usym{k}^{9} &
    \Usym{r} &= \Usym{k}^{-9}
    \\
    \Usym{Q} &= \Usym{k}^{10} &
    \Usym{q} &= \Usym{k}^{-10}
  \end{align*}
  are syntactically forbidden.
\end{ENEW}

\subsubsection{Problem Case: Units, \NEW{Kinds} and Dimensions}
\label{unit-dim-aargh}
  
There is a tendency in both theory and practice to assert a distinguished
reference unit per dimension.  This has a number of consequences, some being
clearly unintended.

It is common for tools to specify conversion factors between commensurable units
in relation to the reference unit of that dimension, both in order to avoid
inconsistencies and as a means of data compression \cite{Allen2004}.  But
dimensions are too coarse to distinguish all semantically incommensurable
quantities, leading to manifestly unsound conversion judgements.  For example,
\emph{newton-meters} (of \emph{torque}) and \emph{joules} (of \emph{work})
become interconvertible in this manner; the same holds for \emph{becquerels} (of
\emph{radioactivity}) and \emph{revolutions-per-minute} (of \emph{angular
  velocity})\footnote{~The GNU command line tool \texttt{units}, as of
  version~2.21, happily suggests the conversion factor $30/\pi$.}.

\begin{ENEW}
  The ISQ admits a concept of \emph{kind of quantity}, such that units of the same
  dimension must additionally be of the same kind in order for their values
  to be commensurable.  Thus, the above problem could be prevented by noting
  that the respective quantities are of different kind.

  However, the standard gives no guidelines how a particular choice of kinds and
  their assignment to quantities can be made, justified or predicated on a
  context.  Instead, it states that ``division into several kinds [\dots] is to
  some extent arbitrary'', and that quantities may be ``generally considered''
  or ``by convention not regarded'' as being of the same kind
  \cite[§3.2]{ISO80000-1}.  Furthermore, there is no indication whether and how
  kinds of quantities should relate to the units used for expressing their
  respective values: Can incommensurable quantities have values involving the
  same units, and if so, what are the appropriate algebraic rules?
\end{ENEW}

\NEW{Logically related, but} even more worrying from the formal perspective is
the ensuing confusion about the very identity of units.  For example, the CGPM
has recognized the medically important, even potentially life-saving, need to
distinguish the units \emph{gray} and \emph{sievert} of \emph{absorbed} and
\emph{equivalent/effective dose}, respectively, of ionizing radiation
\cite{CGPM1979R5}.  Nevertheless, the same resolution ends with the ominous
statement ``The sievert is equal to the joule per kilogram'', where the latter
is the reductionistic definition of the \emph{gray}.  Whatever equality is
invoked in this statement, it is not the usual mathematical one that satisfies
the \emph{indiscernibility of identicals}, seeing as that would utterly defeat
the purpose of said resolution.

\subsection{Research Programme}

The current state of the theory of quantity values and units of measure is at
best semi-formal.  There is a collection of conventions, embodied explicitly in
international standards and implicitly in scientific education, that prescribes
the notations, pronounciations and calculations associated with quantity value
expressions.  What is lacking is an internal method, such as deduction from
axioms, for justification or critique of these rules, in particular where they
give rise to ambiguous, peculiar, awkward or outright contradictory usage.

Our present work is a proposal \NEW{of} how to improve the formal rigor of
reasoning by mathematical remodeling.  An interpretation of the domain of
discourse is given that is \emph{informed}, but not \emph{governed}, by the
traditional readings.  For such a model to be \emph{adequate}, it is required
that the uncontroversial \NEW{consequences} of the theory are upheld.  Where the
model begs to differ from conventional wisdom, it can be studied on an objective
formal basis, opening new and potentially fruitful avenues of criticism.

The result of this research aspires not to produce only \emph{some} model that
satisfies the intended propositions in arbitrary ways.  Rather, the goal is an
\emph{epistemological} model, where the reasons for the truth of some satisfied
proposition can be analyzed.  To this end, a clear line is carefully drawn
between \emph{necessary} and \emph{contingent} properties of the model: The
former are those that follow from the mathematical structure without regard to
the meaning of particular entity symbols, and cannot be subverted.  The latter
are those that represent the actual semantic conventions of science, which
``could be different but aren't''\,---\,i.e., which could be extended or
modified without breaking the underlying logic.

\subsection{Contributions}

The present work, an extended revision of
\cite{DBLP:conf/RelMiCS/WidemannL23,TyWLepper2022preprint}, lays the foundation
for a formal model of dimensions and units of measure that extends and
complements the findings of \cite{Kennedy1996}.  We refine the mathematical
structures outlined there, in order to accomodate unit conversions and prefixes.
In addition, the proposed model rectifies irrelevant, historically accrued
restrictions of traditional notation systems.

In particular, the main contributions of the present work are as follows:
\begin{itemize}
\item A formal mathematical model of dimensions and units with optional
  prefixes, separated into the generic structures shared by all unit systems
  (Section~\ref{structs}) and the properties to be interpreted contingently
  by each unit system (Section~\ref{props}), together with a hierarchy of
  equivalence relations.
\item A formal mathematical model of unit conversion rules as a class of ternary
  relations with a closure operator, that function both as groups and as
  categories (Section~\ref{rels}), together with a hierarchy of six subclasses,
  each \NEW{satisfying stronger useful algebraic properties} than the preceding,
  and in particular an efficient rewriting procedure for calculating conversion
  factors.
\end{itemize}

\section{Prerequisites}

We assume that the reader is familiar with basic concepts of abstract algebra
and category theory.  In particular, the proposed model hinges on the concepts
of \emph{abelian groups} and \emph{adjoint functors}; see \cite{Adamek1990} for
the categorial approach to the former and \cite{MacLane1978} for an introduction
to the pervasive structural role of the latter.

Namely, the promised epistemological distinction shall be realized by modeling
the necessary aspects of data and operations in terms of functors and natural
transformations that arise from adjunctions.  By contrast, contingent aspects
are modeled as particular objects and non-natural morphisms.

\subsection{Abelian Groups}

The category $\mathbf{Ab}$ has abelian groups as objects.  We write $\Forget(G)$
for the carrier set of a group $G$.  (A group and its carrier set are not the
same.)  A homomorphism $f : G \to H$ is a map $f : \Forget(G) \to \Forget(H)$
that commutes with the group operations.  The carrier-set operator $\Forget$,
together with the identity operation on homomorphism maps, $\Forget(f) = f$, is
the forgetful functor from the category $\mathbf{Ab}$ of abelian groups to the
category $\mathbf{Set}$ of sets.

Wherever a generic group variable $G$ is mentioned, we write $\diamond$ for the
group operation, $e$ for the neutral element, and ${}^\dag$ for inversion.
Alternatively, groups could be spelled out as formal tuples of the form
$G = (\Forget(G), {\diamond}, e, {}^\dag)$, but there is no pressing need for
that level of formality.  Powers of group elements are written and understood as
follows:
\begin{align*}
  x^{(0)} &= e & x^{(1)} &= x & x^{(m + n)} &= x^{(m)} \diamond x^{(n)} & x^{(-n)} &= (x^{(n)})^\dag
\end{align*}

In this and other places, the generic model structure makes use of the additive
group of integers~$\mathbb{Z}$.  Contingent actual numbers in some model
instance are from the set $\RatP$ of positive rationals.  For arithmetics only
the multiplicative group $\RatM$ on that set is required.  We refer to these
numbers as \emph{ratios}.

Kennedy \cite{Kennedy1996} observed that the algebras of dimensions and units
are essentially \emph{free abelian groups}.

\subsection{Free Abelian Groups}

We write $\Abel(X)$ for the free abelian group over a set $X$ of generators.
Whereas the concept is specified only up to isomorphism in category theory, we
conceive of a particular construction of the carrier set $\Forget(\Abel(X))$,
namely the space $\FinSupp(X)$ of finitely supported integer-valued functions on
$X$:
\begin{align*}
  \FinSupp(X) &= \{ f : X \to \mathbb{Z} \mid \operatorname{supp}(f) \enspace\text{finite} \}
  &
  \operatorname{supp}(f) &= \{ x \in X \mid f(x) \neq 0 \}
\end{align*}

If $X$ is finite, then the group $\Abel(X)$ is called \emph{finitely generated}.
Note that the finite-support constraint is redundant in this case.

For denoting a particular element of a free abelian group directly, it suffices
to refer to the support.  The finite partial maps of type
$f : X \to \mathbb{Z} \setminus \{ 0 \}$ are in one-to-one correspondence with
group elements $f\subzero \in \FinSupp(X)$:
\begin{align*}
  f\subzero(x) &=
  \begin{cases}
    f(x) & \text{if defined}
    \\
    0 & \text{otherwise}
  \end{cases}
\end{align*}

\begin{ENEW}
  For example, consider the positive rational numbers $\mathbb{Q}_+$, which by
  the fundamental theorem of arithmetic are isomorphic to the free abelian group
  over the prime numbers $\mathbb{P}$.  Namely, the map
  $^\star : \FinSupp(\mathbb{P}) \to \mathbb{Q}_+$ that sends every finitely
  supported integer-valued map on the primes to the corresponding product of
  prime powers,
  \begin{align*}
    f^\star &= \prod_{p \in P} p^{f(p)}
  \end{align*}
  is a bijection.  Then we can express the fact that $2/9$ factors as
  $2 \cdot 3^{-2}$ concisely as follows:
  \begin{align*}
    \frac29 &= {\{ 2 \mapsto 1, 3 \mapsto -2 \}\subzero}^\star
  \end{align*}
\end{ENEW}

The operation that turns the set $\FinSupp(X)$ into the abelian group $\Abel(X)$
is realized by pointwise addition.  However, since the standard interpretation
of function values is \emph{power exponents}, we shall use multiplicative
notation, such that $(fg)(x) = f(x) + g(x)$, with neutral element
$1 = \varnothing\subzero$.

For the following definitions, it is useful to introduce the \emph{Iverson
  bracket}:
\begin{align*}
  [p] &= \begin{cases}
    1 & \text{if}~p~\text{holds}
    \\
    0 & \text{otherwise}
  \end{cases}
\end{align*}

For one, $\Abel$ is a functor
\begin{align*}
  \Abel(f : X \to Y)(g \in \FinSupp(X))(y \in Y) &= \sum_{x \in X} [f(x) = y] \cdot g(x)
\end{align*}
that accumulates the contents of a group element $g$ over the kernel of the map
$f$.
Furthermore, $\Abel$ is left adjoint to $\Forget$: There are two natural
transformations
\begin{align*}
  \delta &: 1 \Rightarrow \Forget\Abel & \varepsilon &: \Abel\Forget \Rightarrow 1
\end{align*}
in $\mathbf{Set}$ and $\mathbf{Ab}$, named the
\emph{unit}\footnote{Adjunction/monad units are unrelated to units of measure.}
and \emph{counit} of the adjunction, respectively, such that:
\begin{align*}
  \varepsilon\Abel \circ \Abel\delta &= 1\Abel & \Forget\varepsilon \circ \delta\Forget &= 1\Forget
\end{align*}

Their concise definitions are:
\begin{align*}
  \delta_X(x) &= \{ x \mapsto 1 \}\subzero
  & \varepsilon_G(f \in \Abel\Forget(G)) &= \mathop\lozenge_{x \in \Forget(G)} x^{(f(x))}
  \\
              &= [x = \_]
\end{align*}
Note how the finite-support property of $f$ ensures that the infinitary
operator $\lozenge$ is well-defined.

\begin{figure}
  \begin{ENOW}
  \begin{align*}
    \xymatrix{
      \FinSupp
      \ar[r]^{\delta\FinSupp}
      \ar[d]_{\FinSupp\delta}
      \ar@{=}[rd]
      &
      \FinSupp^2
      \ar[d]^{\lambda}
      \\
      \FinSupp^2
      \ar[r]_{\lambda}
      &
      \FinSupp
    }
    &&
    \xymatrix{
      \FinSupp^3
      \ar[r]^{\FinSupp\lambda}
      \ar[d]_{\lambda\FinSupp}
      &
      \FinSupp^2
      \ar[d]^{\lambda}
      \\
      \FinSupp^2
      \ar[r]_{\lambda}
      &
      \FinSupp
    }
  \end{align*}
  \caption{Laws of the Monad $\FinSupp$}
  \label{fig:abel-monad}
\end{ENOW}
\end{figure}

From the adjoint situation follows that the composite functor
$\FinSupp = \Forget\Abel$ is a monad on $\mathbf{Set}$, with \emph{unit}
$\delta$ and \emph{multiplication} $\lambda = \Forget\varepsilon\Abel$:
\begin{align*}
  \delta &: 1 \Rightarrow \FinSupp
  \\
  \lambda &: \FinSupp^2 \Rightarrow \FinSupp & \lambda_X(f)(x) &= \sum_{\mathclap{g \in \FinSupp(X)}} f(g) \cdot g(x)
\end{align*}

Monads and their operations are often best understood as abstract syntax: The
elements of $\FinSupp(X)$ encode group words, up to necessary equivalence.  The
unit $\delta$ embeds the generators into the words as \emph{literals}.  The
multiplication $\lambda$ \emph{flattens} two \emph{nested} layers of words.
The counit $\varepsilon$ embodies the more general concept of interpreting
words over a particular target group by ``multiplying things out''.

\begin{lemma}
  \label{lemma:powers}
  Any element of a free abelian group can be written, uniquely up to
  permutation of factors, as a finite product of nonzeroth powers:
  \begin{align*}
    \{ x_1 \mapsto z_1, \dots, x_n \mapsto z_n \}\subzero &= \prod_{i=1}^n \delta(x_i)^{z_i}
  \end{align*}
\end{lemma}

\subsection{Pairs}

Both the Cartesian product of two sets and the direct sum of two abelian groups
are instances of the categorial binary product.  All of their relevant
structures and operations arise from a \emph{diagonal functor}
\begin{align*}
  \Delta_C &: C \to C \times C & \Delta_C(X) &= (X, X) & \Delta_C(f : X \to Y) = (f, f)
\end{align*}
into the product category of some \NEW{suitable} category $C$ with itself, and
the fact that it has a right adjoint $(\_ \times \_)$.  Namely, there are unit
and counit natural transformations:
\begin{align*}
  \psi &: 1 \Rightarrow (\_ \times \_)\Delta & \pi &: \Delta(\_ \times \_) \Rightarrow 1
  \\
  \psi_X(x) &= (x, x) & \pi_{(X, Y)} &= (\pi_1, \pi_2) \quad\text{where}\quad
  \begin{aligned}
    \pi_1(x, y) &= x
    \\
    \pi_2(x, y) &= y
  \end{aligned}
\end{align*}

\begin{figure}
  \begin{ENOW}
  \begin{align*}
    \xymatrix{
      && Y
      \\
      X
      \ar[rru]^{f}
      \ar[rrd]_{g}
      \ar@{-->}[rr]^{\exists! \langle f, g \rangle}
      &&
      Y \times Z
      \ar[u]_{\pi_1}
      \ar[d]^{\pi_2}
      \\
      && Z
    }
  \end{align*}
  \caption{Cartesian Pair Operations}
  \label{fig:pairs}
\end{ENOW}
\end{figure}

In traditional notation, the counit components (projections) $\pi_1$ and $\pi_2$
are used separately, whereas the pairing of morphisms is written with angled
brackets:
\begin{align*}
  \langle f : X \to Y, g : X \to Z \rangle &= (f \times g) \circ \psi_X : X \to (Y \times Z)
\end{align*}

\subsection{Pairing With an Abelian Group}

Consider the functor that pairs elements of arbitrary sets $X$ with elements of
(the carrier of) a fixed abelian group $G$:
\begin{align*}
  \Pair[G](X) &= \Forget(G) \times X & \Pair[G](f) &= \mathrm{id}_{\Forget(G)} \times f
\end{align*}

From the group structure of $G$ (or actually any monoid), we obtain a
simple monad, with unit $\eta_G$ and multiplication $\mu_G$:
\begin{align*}
  \eta_G &: 1 \Rightarrow \Pair[G] & \eta_{G,X}(x) &= (e, x)
  \\
  \mu_G &: {\Pair[G]}^2 \Rightarrow \Pair[G] & \mu_{G,X}\bigl(a, (b, x)\bigr) &= (a \diamond b, x)
\end{align*}

When the second component is changed to abelian groups as well, a corresponding
functor on $\mathbf{Ab}$ is obtained, which creates the direct sum of abelian
groups instead of the Cartesian product of carrier sets:
\begin{align*}
  \GPair[G](H) &= G \times H & \GPair[G](f) &= \mathrm{id}_G \times f
\end{align*}

The two monads are related by the forgetful functor:
\begin{align*}
  \Pair[G] \Forget &= \Forget \GPair[G]
\end{align*}

For the present discussion, there is no need to spell out particular
adjunctions that give rise to these monads, nor their counits.

\subsection{Monad Composition}

There is a natural transformation in $\mathbf{Ab}$ that exchanges pairing and
free abelian group construction:
\begin{align*}
  \beta_G &: \Abel\, \Pair[G] \Rightarrow \GPair[G]\,\Abel & \beta_{G,X} &= \bigl\langle \varepsilon_G \circ \Abel(\pi_1), \Abel(\pi_2) \bigr\rangle
\end{align*}

\begin{figure}
  \begin{ENOW}
  \begin{equation*}
    \xymatrix{
      &
      \Abel\bigl(\Forget(G)\bigr)
      \ar[rr]^{\varepsilon_G}
      &&
      G
      \\
      \Abel\bigl(\Pair[G](X)\bigr)
      \ar@{=}[r]
      &
      \Abel\bigl(\Forget(G) \times X\bigr)
      \ar[u]^{\Abel(\pi_1)}
      \ar[rd]_{\Abel(\pi_2)}
      \ar@{-->}[rr]^{\exists!\beta_G}
      &&
      G \times \Abel(X)
      \ar@{=}[r]
      \ar[u]^{\pi_1}
      \ar[ld]_{\pi_2}
      &
      \GPair[G]\bigl(\Abel(X)\bigr)
      \\
      &&
      \Abel(X)
    }
  \end{equation*}
  \caption{Construction of $\beta$}
  \label{fig:beta}
\end{ENOW}
\end{figure}

Back down in $\mathbf{Set}$, this is a natural transformation
\begin{align*}
  \Forget(\beta_G) = \beta_G : \FinSupp\, \Pair[G] \Rightarrow \Pair[G]\, \FinSupp
\end{align*}
which can be shown to be a distributive law between the monads $\Pair[G]$ and
$\FinSupp$:
\begin{align*}
  \beta_G \circ \delta \Pair[G] &= \Pair[G] \delta
  & \beta_G \circ \FinSupp \mu_G &= \lambda \FinSupp \circ \Pair[G] \beta_G \circ \beta_G \Pair[G]
  \\
  \beta_G \circ \FinSupp \eta &= \eta \FinSupp
  & \beta_G \circ \lambda \Pair[G] &= \Pair[G] \mu_G \circ \beta_G \FinSupp \circ \FinSupp \beta_G
\end{align*}

\begin{figure}
  \begin{ENOW}
  \begin{align*}
    \xymatrix{
      &
      \FinSupp
      \ar[ld]_{\FinSupp\eta_G}
      \ar[rd]^{\eta_G\FinSupp}
      \\
      \FinSupp\,\Pair[G]
      \ar[rr]_{\beta_G}
      &&
      \Pair[G]\FinSupp
      \\
      &
      \Pair[G]
      \ar[lu]^{\delta\Pair[G]}
      \ar[ru]_{\Pair[G]\delta}
    }
    &&
    \xymatrix{
      \FinSupp\,\Pair[G]^2
      \ar[r]^{\beta_G\Pair[G]}
      \ar[d]_{\FinSupp\mu_G}
      &
      \Pair[G]\FinSupp\,\Pair[G]
      \ar[r]^{\Pair[G]\beta_G}
      &
      \Pair[G]^2\FinSupp
      \ar[d]^{\mu_G\FinSupp}
      \\
      \FinSupp\,\Pair[G]
      \ar[rr]_{\beta_G}
      &&
      \Pair[G]\FinSupp
      \\
      \FinSupp^2\Pair[G]
      \ar[r]_{\FinSupp\beta_G}
      \ar[u]^{\lambda\Pair[G]}
      &
      \FinSupp\,\Pair[G]\FinSupp
      \ar[r]_{\beta_G\FinSupp}
      &
      \Pair[G]\FinSupp^2
      \ar[u]_{\Pair[G]\lambda}
    }
  \end{align*}
  \caption{$\beta$ is a Distributive Law}
\end{ENOW}
\end{figure}

This turns the composite functor $\Pair[G] \FinSupp$ into a monad,
with unit $\theta_G$ and multiplication $\xi_G$:
\begin{align*}
  \theta_G &: 1 \Rightarrow \Pair[G]\, \FinSupp & \theta_G &= \eta_G\delta
  \\
  \xi_G &: (\Pair[G]\, \FinSupp)^2 \Rightarrow \Pair[G]\, \FinSupp & \xi_G &= \mu_G\lambda \circ \Pair[G]\beta_G\FinSupp
\end{align*}

\begin{figure}
  \begin{ENOW}
  \begin{align*}
    \xymatrix{
      &
      \FinSupp
      \ar[rd]^{\eta_G\FinSupp}
      \\
      1
      \ar[ru]^{\delta}
      \ar[rd]_{\eta_G}
      \ar@{-->}[rr]^{\theta_G}
      &&
      \Pair[G]\FinSupp
      \\
      &
      \Pair[G]
      \ar[ru]_{\Pair[G]\delta}
    }
    &&
    \xymatrix{
      &&&
      \Pair[G]^2\FinSupp
      \ar[rd]^{\mu_G\FinSupp}
      \\
      (\Pair[G]\FinSupp)^2
      \ar[rr]^{\Pair[G]\beta_G\FinSupp}
      &&
      \Pair[G]^2\FinSupp^2
      \ar[ru]^{\Pair[G]^2\lambda}
      \ar[rd]_{\mu_G\FinSupp^2}
      \ar@{-->}[rr]^{\xi_G}
      &&
      \Pair[G]\FinSupp
      \\
      &&&
      \Pair[G]\FinSupp^2
      \ar[ru]_{\Pair[G]\lambda}
    }
  \end{align*}
  \caption{Construction of Composite Monad $\Pair[G]\FinSupp$}
\end{ENOW}
\end{figure}

\subsection{Concrete Examples}

The structures and operations introduced so far are all natural in the
general categorial sense, but range from quite simple to rather
nontrivial in complexity.  To illustrate this, the following section
gives concrete examples with actual elements.

In all of the following formulae, let $X = \{ a, b, c, d, \dots \}$.

Consider two elements $g_1, g_2 \in \FinSupp(X)$.
Thanks to the finite map notation, their support is clearly visible:
\begin{align*}
  g_1 &= \{ a \mapsto 2, b \mapsto -1, c \mapsto 1 \}\subzero
  &
    g_2 &= \{ b \mapsto 2, c \mapsto -1, d \mapsto -2 \}\subzero
  \\
  \operatorname{supp}(g_1) &= \{ a, b, c \}
  &
  \operatorname{supp}(g_2) &= \{ b, c, d \}
\end{align*}

The group operation acts by pointwise addition, zeroes vanish in the notation:
\begin{align*}
  g_1g_2 &= \{ a \mapsto 2, b \mapsto 1, d \mapsto -2 \}\subzero
\end{align*}

There is no general relationship between $\operatorname{supp}(g_1)$,
$\operatorname{supp}(g_2)$ and $\operatorname{supp}(g_1g_2)$, because of
cancellation.
Now consider a map $f : X \to X$, such that:
\begin{align*}
  f(a) = f(b) = a && f(c) = f(d) = b
\end{align*}

Mapping $f$ over group elements yields
\begin{align*}
  \begin{aligned}
    \Abel(f)(g_1) &= \{ a \mapsto 1, b \mapsto 1 \}\subzero
    \\
    \Abel(f)(g_2) &= \{ a \mapsto 2, b \mapsto -3 \}\subzero
  \end{aligned}
  &&
  \Abel(f)(g_1g_2) &= \{ a \mapsto 3, b \mapsto -2 \}\subzero = \Abel(f)(g_1)\, \Abel(f)(g_2)
\end{align*}
where the coefficients of identified elements are lumped together.
Note that $\Abel(f)(g)$ is very different from (the contravariant) $g \circ f$.
Monad multiplication takes integer-valued maps of integer-valued maps,
distributes coefficients, and combines the results:
\begin{align*}
  \lambda_X\bigl(\{ g_1 \mapsto -2, g_2 \mapsto -1 \}\subzero\bigr) &= g_1^{-2}g_2^{-1} =
  \{ a \mapsto -4, c \mapsto -1, d \mapsto 2 \}\subzero
\end{align*}

The counit takes an integer-valued map of group elements, and evaluates it as a
group expression:
\begin{align*}
  \varepsilon_\RatM\bigl(\underbrace{\{ 2 \mapsto -3, 3 \mapsto 2, \tfrac25 \mapsto -1 \}}_{q_1}\bigr) &= 2^{-3}\, 3^2\, \bigl(\tfrac25\bigr)^{-1} = \frac{45}{16}
\end{align*}

Note how the integer coefficients in $\FinSupp(X)$ function as exponents when
the target group is multiplicative.
The operations of the pairing functor $\Pair[G]$ are straightforward:
\begin{align*}
  \eta_{\RatM,X}(a) &= (1, a)
  &
  \mu_{\RatM,X}\bigl(2, (3, b)\bigr) &= (2 \cdot 3, b) = (6, b)
  &
  \Pair[\RatM](f)(5, c) &= \bigl(5, f(c)\bigr) = (5, b)
\end{align*}

By contrast, the distributive law $\beta$ functions in a fairly complex way,
splitting the left hand sides of a group element map:
\begin{align*}
  \beta_{\RatM, X}\bigl(\bigl\{ (2, a) \mapsto -3, (3, b) \mapsto 2, \bigl(\tfrac25, c\bigr) \mapsto -1 \bigr\}\bigr)
  &=
  \bigl(\varepsilon_\RatM(q_1), \underbrace{\{ a \mapsto -3, b \mapsto 2, c \mapsto -1 \}}_{g_3}\bigr)
\end{align*}

The operations of the composite monad $\Pair[G]\FinSupp$ just combine the above.
\begin{align*}
  \theta_\RatM(c) &= \bigl(1, \{ c \mapsto 1\}\subzero \bigr)
  \\[1ex]
  &\hphantom{{}={}} \xi_{\RatM, X}\Bigl(2, \bigl\{ \bigl(3, \{ a \mapsto 5 \}\subzero\bigr) \mapsto -2, \bigl(7, \{ b \mapsto -1 \}\subzero\bigr) \mapsto 1\bigr\}\subzero\Bigr)
  \\
  &= \bigl(2 \cdot 3^{-2} \cdot 7^1, \{ a \mapsto 5 \cdot -2, b \mapsto -1 \cdot 1 \}\subzero\bigr)
  \\
  &= \bigl(\tfrac{14}{9}, \{ a \mapsto -10, b \mapsto -1 \}\subzero\bigr)
\end{align*}

\section{Data Structures}

\subsection{Generic Structures}
\label{structs}

This section defines the structures of a formal model of dimensions and units.
These are \emph{generic}: no interpretation of base symbols is presupposed.
This is achieved by $\delta$, $\lambda$, $\varepsilon$, $\pi$, $\psi$,
$\eta$, $\mu$, $\beta$, $\theta$, $\xi$ all being \emph{natural
  transformations}: They are parametric families of maps that transform data
between \emph{shapes} specified by functors, in a way that is logically
independent of the elementary \emph{content} specified by the type argument.

\begin{definition}[Dimension]
  \label{def:dim}
  The \emph{dimensions} $\Dim$ are the free abelian group over $\Dim\subbase$, a
  given finite set of \emph{base dimensions}:
  \begin{align*}
    \Dim &= \Abel(\Dim\subbase)
  \end{align*}
\end{definition}

\begin{example}[SI Dimensions]
  The SI/ISQ recognizes seven physical base dimensions,
  \begin{math}
    \Dim\subbase^{\mathrm{SI}} = \{ \mathsf{L}, \mathsf{T}, \mathsf{M}, \mathsf{I}, \mathsf{\Theta}, \mathsf{N}, \mathsf{J} \}
  \end{math}.   From these, compound dimensions can be
  formed; for example, quantities of \emph{thermodynamical entropy} are associated with
  \begin{math}
    d_H = \{ \mathsf{L} \mapsto 2, \mathsf{M} \mapsto 1, \mathsf{T} \mapsto -2, \mathsf{\Theta} \mapsto -1 \}\subzero
  \end{math}.  However, the actual physical interpretation of these symbols need
  not be considered at all for formal analysis.
\end{example}

\begin{remark} We do not use the product-of-powers notation of
  Lemma~\ref{lemma:powers} for concrete elements.  Neither do we recommend that
  use for formal analysis, because it is prone to causing subtle
  misunderstandings, in particular the highly overloaded expression for the
  empty case, $1 = \varnothing\subzero$.  Confer Section~\ref{1-aargh}.
\end{remark}

\begin{definition}[Root Unit]
  \label{def:root}
  The \emph{root units} $\Unit\subroot$ are the free abelian group over
  $\Unit\subbase$, a given finite set of \emph{base units}:
  \begin{align*}
    \Unit\subroot &= \Abel(\Unit\subbase)
  \end{align*}
  \NEW{Thus, root units have a structure analogous to dimensions, compare
    Definition~\ref{def:dim}.  The reason that they are called ``root units'' and
    not just ``units'' is that prefixes are not yet accounted for; see
    Definition~\ref{def:pref} below.}
\end{definition}

\begin{ENOW}
  \begin{example}[Base Units and the SI]
    We use the qualifier \emph{base} always in a strictly syntactic sense, namely
    to denote the atomic symbols over which compound expressions, or group
    words, can be formed.  By contrast, the SI/ISQ attaches to it a pragmatic
    meaning, namely to denote entities that require no definition by reduction
    to other, more basic entities.  While these contending concepts overlap to
    some degree, there are notable caveats and exceptions in both directions.

    On the one hand, \cite{ISO80000-1} distinguishes one base unit per base
    dimension.  Six of these are also base units in the formal syntactic
    sense.  The seventh however, for historical reasons, is not: The
    distinguished unit of \emph{mass} is the \emph{kilogram} ($\mathrm{kg}$),
    which is syntactically a compound derived from the simpler unit \emph{gram}
    ($\mathrm{g}$); see Example~\ref{ex:kilogram} below.
    \begin{equation*}
      \Unit\subbase^{\mathrm{SI}_0} = \{ \mathrm{m}, \mathrm{s}, \mathrm{kg}, \mathrm{A}, \mathrm{K}, \mathrm{mol}, \mathrm{cd} \}
    \end{equation*}

    On the other hand, SI units that have both a dedicated symbol and a
    definition are not considered base units in the pragmatic sense, but
    \emph{derived units with special names}.  For example, the \emph{ohm} is
    is associated with both the symbol $\mathrm{\Omega}$ and the defining term
    $\mathrm{kg}{\cdot}\mathrm{m}^2{\cdot}\mathrm{s}^{-3}{\cdot}\mathrm{A}^{-2}$.
    
    The devil is in the details: Before the major revision of SI definitions by
    the 26th CGPM in 2018, the SI standard stated that ``the ohm, symbol
    $\mathrm{\Omega}$, is uniquely defined by the relation
    $\mathrm{\Omega} =
    \mathrm{kg}{\cdot}\mathrm{m}^2{\cdot}\mathrm{s}^{-3}{\cdot}\mathrm{A}^{-2}$''\;\cite[§2.1.1]{sib8},
    implying that the two expressions are distinct, relatable entities.  By
    contrast, the statement has been dropped from the later version \cite{sib9}.
    There, the pragmatic foundation is discussed thus:

    \begin{quote}
      ``[\dots] this distinction is, in principle, not needed, since all units,
      \emph{base} as well as \emph{derived units}, may be constructed directly
      from the defining constants.  Nevertheless, the concept of base and
      derived units is maintained because it is useful and historically well
      established [\dots]''\;\cite[§2.3]{sib9}
    \end{quote}
    Compatibility with the ISQ, where base units are in one-to-one
    correspondence with base quantities and hence base dimensions, is also
    cited.  See Example~\ref{ex:si-well} and Section~\ref{si-co} for a
    formal discussion of the traditional base--derived dichotomy in the SI/ISQ.

    Since the explanatory power of our model depends essentially on a sharp
    distinction between syntax and semantics, and a full abstraction from
    pragmatics, we hold that $\mathrm{\Omega}$ is a base unit, whereas
    $\mathrm{kg}{\cdot}\mathrm{m}^2{\cdot}\mathrm{s}^{-3}{\cdot}\mathrm{A}^{-2}$
    is a (derived/compound) root unit, and that the two can and must be related
    explicitly.
  \end{example}
\end{ENOW}

\begin{ENOW}
  \begin{example}[Root Units and SI Coherence]
    \label{ex:root-si}
    The concept of root units is not explicit in the SI/ISQ, but a construct of
    our remodeling.  Again, it overlaps partially with a concept found there,
    namely \emph{coherent units}.  And again, the differences are mostly due to
    a distinction between syntax and semantics.

    Coherence has two alternative characterizations (semiformal definitions),
    and standards disagree about their relative logical precedence; the SI gives
    the first as a definition and the second as a corollary \cite[§2.3.4]{sib9},
    whereas the ISQ has them in reverse \cite[§6]{ISO80000-1}:

    In the first, semantic approach, defined units are considered coherent, if
    their definition does not imply a numerical factor other than one.  See
    Definition~\ref{def:val-base} and Remark~\ref{rem:conv} below for two
    distinct semantic sources of numerical factors.

    In the second, syntactic approach, defined units are be considered coherent
    if the form of their definition has the same shape as the underlying
    compound quantities or dimensions.  This notion can be formalized: Define a
    function $u : \Dim\subbase \to \Unit\subbase$ that chooses a distinguished
    base unit per base dimension, such as in \cite[§6.5.3]{ISO80000-1}.  Let
    $d \in \Dim$ be the dimension associated with a quantity $q$.  Then the
    coherent (root) unit for $q$ is $\Abel(u)(d)$.
    
    Unfortunately the latter, seemingly syntactic definition is clearly meant
    only up to semantic equivalence in the SI/ISQ: For example, the unit symbol
    $\mathrm{\Omega}$ is listed as coherent, but is clearly not of the proposed
    form, whereas the representation of its intended definition as a root unit,
    \begin{equation*}
      \{ \mathrm{kg} \mapsto 1, \mathrm{m} \mapsto 2, \mathrm{s} \mapsto -3, \mathrm{A} \mapsto -2 \}\subzero
    \end{equation*}
    
    \eject
    
    \noindent
    matches the shape of the dimension of the quantity \emph{electric
      resistance}
    \begin{equation*}
      \{ \mathsf{M} \mapsto 1, \mathsf{L} \mapsto 2, \mathsf{T} \mapsto -3, \mathsf{I} \mapsto -2 \}\subzero
    \end{equation*}
    perfectly, given $u(\mathrm{kg}) = \mathsf{M}$ etc.

    In order to disentangle syntax and semantics in our model, we shall use the
    novel concept of root units for the syntactic characterization, which is
    narrower than traditional coherence, and remodel the semantic aspect later;
    see Definition~\ref{def:unit-coherent} below.
    
  \end{example}
\end{ENOW}

\begin{definition}[Prefix]
  \label{def:pref}
  The \emph{prefixes} $\Pref$ are the free abelian group over
  $\Pref\subbase$, a given finite set of \emph{base prefixes}:
  \begin{align*}
    \Pref &= \Abel(\Pref\subbase)
  \end{align*}
\end{definition}

\begin{example}[SI Prefixes]
  \label{ex:si-pref}
  Three families of prefixes are recognized in combination with the SI units.
  The symbols come with associated numerical values, discussed in detail in
  Section~\ref{props} below.  For now it suffices to simply name
  them:
  \begin{align*}
    \Pref\subbase^{\mathrm{SI}} &= \left\{
      \begin{array}{l}
        \mathrm{d}, \mathrm{c}, \mathrm{m}, \mathrm{\upmu}, \mathrm{n}, \mathrm{p}, \mathrm{f}, \mathrm{a}, \mathrm{z}, \mathrm{y}, \mathrm{r}, \mathrm{q},
        \\
        \mathrm{da}, \mathrm{h}, \mathrm{k}, \mathrm{M}, \mathrm{G}, \mathrm{T}, \mathrm{P}, \mathrm{E}, \mathrm{Z}, \mathrm{Y}, \mathrm{R}, \mathrm{Q},
        \\
        \mathrm{ki}, \mathrm{Mi}, \mathrm{Gi}, \mathrm{Ti}, \mathrm{Pi}, \mathrm{Ei}, \mathrm{Zi}, \mathrm{Yi}
      \end{array}
    \right\}
  \end{align*}
  The symbol $\mathrm{m}$ occurs as both a base prefix and a base unit in the SI
  vocabulary.  However, this causes ambiguity issues only for parsing
  traditional notations, where prefix and unit symbols are simply concatenated:
  \NEW{For example, $\Usym{m}\hskip 1pt\Usym{s}^{-1}$ could parse either as
    \emph{meter per second} or as \emph{per millisecond}; see
    \cite[§7.2.2]{ISO80000-1}.} The formal semantics discussed here are not
  affected.  See \NEW{also} Example~\ref{ex:precip} below for unambiguous usage.
\end{example}

\begin{remark}
  Combinations of multiple base prefixes are forbidden in modern scientific
  notation.  However, there are good reasons to deviate from that rule:
  \begin{itemize}
  \item The empty prefix is always allowed, and modeled adequately as
    $\varnothing\subzero$.
  \item Some traditional usages of \emph{double} prefixes are known, and
    double prefixes beyond \emph{quetta-} have recently been resuggested
    \cite{BROWN2019237}.
  \item The inner logic of each SI prefix family is a geometric sequence,
    confer Section~\ref{prefix-aargh} \NEW{and Example~\ref{ex:si-val} below}.
  \item Last but not least, composite prefixes arise virtually in the algebra of
    composite units.  \NEW{For example, semantic analysis of the
      \emph{centilitre} discussed in Section~\ref{eq-aargh} by reduction to the
      coherent SI unit of volume, the \emph{cubic meter}, involves the effective
      composite prefix $\Usym{c\,d^3}$.  For a formalization of the procedure,
      see Definitions~\ref{def:subnorm} and \ref{def:norm} below.}
  \end{itemize}
  Thus, the generalization of prefixes to the free abelian group is a
  theoretical simplification and unification.
\end{remark}

\begin{definition}[Preunit]
  The \emph{preunits} $\Unit\subpre$ are the Cartesian pairs of a prefix,
  forgetting the group structure, and a base unit:
  \begin{align*}
    \Unit\subpre &= \Pair[\Pref](\Unit\subbase)
  \end{align*}
\end{definition}

\begin{example}
  \label{ex:kilogram}
  The SI unit \emph{kilogram}, already discussed above, is represented formally
  as a preunit,
  \begin{math}
    \mathrm{kg} = \bigl( \delta_{\Pref\subbase}(\mathrm{k}), \mathrm{g} \bigr)
  \end{math}.
\end{example}

\begin{definition}[Unit]
  The \emph{units} $\Unit$ are the free abelian group over the preunits:
  \begin{align*}
    \Unit &= \Abel(\Unit\subpre)
  \end{align*}
\end{definition}

\begin{remark}
  Unlike in the preceding constructions, the generator set $\Unit\subpre$ is
  generally infinite.  Fortunately, this causes no problems for the remainder of
  the present work; see Theorem~\ref{thm:closed-fg} below.
\end{remark}

\begin{definition}[\NEW{Embedding of Base Units}]
Since the structure $\Unit$ arises from the composition of two monadic functors,
the composition of monad unit maps is bound to occur often in its use.  We write
just the bracket $\de\_ : \Unit\subbase \to \Unit$ for \NEW{the horizontal
  composition of unit transformations, i.e.,} the precise but verbose map:
\begin{equation*}
  (\delta\eta_{\Pref})_{\Unit\subbase} = \delta_{\Unit\subpre} \circ \eta_{\Pref, \Unit\subbase}
  = \FinSupp(\eta_{\Pref, \Unit\subbase}) \circ \delta_{\Unit\subbase}
\end{equation*}
\NEW{More concretely, this boils down to the construction:}
\begin{equation*}
  \NEW{\de{x} = \{ (\varnothing\subzero, x) \mapsto 1 \}\subzero}
\end{equation*}
\end{definition}

In other occurrences, the subscripts of $\delta$ and $\eta$ may be
omitted in applications, and can be inferred from the context.

\begin{example}
  \label{ex:newton}%
  There are additional so-called \emph{derived} units in the SI, which have a
  base unit symbol of their own, but whose semantics are defined by reduction to
  other (composite) units, for example the \emph{newton},
  \begin{math}
    \de{\mathrm{N}} \equiv \{ \mathrm{kg} \mapsto 1, \eta(\mathrm{m}) \mapsto 1, \eta(\mathrm{s}) \mapsto -2 \}\subzero
  \end{math}, where $\equiv$ is some semantic equivalence relation yet to be
  specified; see Definitions~\ref{def:norm-equiv}, \ref{def:eval-equiv},
  \ref{def:unit-coherent} and \ref{def:well-def} below.
\end{example}

\begin{definition}[Prefix/Root of a Unit]
  Any unit can be decomposed into a prefix and a root unit, and the root units
  embedded back into the units, by means of natural group homomorphisms:
  \begin{align*}
    \operatorname{pref} &: \Unit \to \Pref & \operatorname{pref} &= \pi_1 \circ \beta_{\Pref,\Unit\subbase} = \varepsilon_\Pref \circ \Abel(\pi_1)
    \\
    \operatorname{root} &: \Unit \to \Unit\subroot & \operatorname{root} &= \pi_2 \circ \beta_{\Pref,\Unit\subbase} = \Abel(\pi_2)
    \\
    \operatorname{unroot} &: \Unit\subroot \to \Unit & \operatorname{unroot} &= \Abel(\eta_{\Pref,\Unit\subbase})
    \\
    \operatorname{strip} &: \Unit \to \Unit & \operatorname{strip} &= \operatorname{unroot} \circ \operatorname{root}
  \end{align*}
\end{definition}

\begin{lemma}
  \label{lemma:strip-idemp}
  The function $\operatorname{root}$ is the left inverse of
  $\operatorname{unroot}$.  Their flipped composition, the function
  $\operatorname{strip}$, is idempotent.
\end{lemma}

\begin{shortproof} \leavevmode
\vspace{-3ex}
  \begin{gather*}
    \operatorname{root} \circ \operatorname{unroot} = \Abel(\pi_2) \circ \Abel(\eta_{\Pref,\Unit\subbase})
    = \Abel(\pi_2 \circ \eta_{\Pref,\Unit\subbase})
    = \Abel(\mathrm{id}_{\Unit\subbase})
    = \mathrm{id}_{\Unit\subroot}
    \\
    \operatorname{strip} \circ \operatorname{strip} = \operatorname{unroot} \circ \underbrace{\operatorname{root} \circ \operatorname{unroot}}_{\mathrm{id}_{\Unit\subroot}} \circ \operatorname{root}
    = \operatorname{unroot} \circ \operatorname{root} = \operatorname{strip}
  \end{gather*}
  \vspace{-3ex}
  
  {~}
\end{shortproof}

\begin{definition}[Root Equivalence]
  \label{def:root-equiv}
  Two units are called \emph{root equivalent}, written $\simeq\subroot$, iff
  their roots coincide:
  \begin{align*}
    u \simeq\subroot v \iff \operatorname{root}(u) = \operatorname{root}(v)
  \end{align*}
\end{definition}

The units as defined above are a faithful semantic model of traditional
notations, confer \cite{sib9}.  However, an algebraically more
well-behaved structure can be derived by transposing prefixes
and free abelian group construction:

\begin{definition}[Normalized Unit]
  \label{def:subnorm}
  The \emph{normalized units} $\Unit\subnorm$ are the direct sum of a
  prefix and a root unit:
  \begin{align*}
    \Unit\subnorm &= \GPair[\Pref](\Unit\subroot)
  \end{align*}
\end{definition}

\begin{definition}[Normalization]
  \label{def:norm}
  The normalized units are derived naturally from the units proper, in terms of
  the natural group homomorphism $\beta$:
  \begin{align*}
    \operatorname{norm} &: \Unit \to \Unit\subnorm & \operatorname{norm} &= \beta_{\Pref,\Unit\subbase} = \langle \operatorname{pref}, \operatorname{root} \rangle
  \end{align*}
\end{definition}

The composite monad structure of
$\Forget(\Unit\subnorm) = \Pair[\Pref]\FinSupp(\Unit\subbase)$ comes with a
multiplication $\xi_{\Pref,\Unit\subbase}$, and thus provides exactly the
compositionality lacking in the traditional model $\Unit$: In a normalized unit,
other normalized units can be substituted for base units and just ``multiplied
out''.  In $\Unit$, this is formally impossible:

\begin{theorem}
  \label{thm:no-strength}
  The group $\Unit = \FinSupp \Pair[\Pref](\Unit\subbase)$ cannot be given a
  composite monadic structure in the same way as $\Unit\subnorm$; except for
  degenerate cases, there is no distributive law
  $\Pair[G]\, \FinSupp \Rightarrow \FinSupp\, \Pair[G]$ that is also a group
  homomorphism.
\end{theorem}

\begin{proof}
  A distributive law $\tau_G : \Pair[G]\,\FinSupp \Rightarrow \FinSupp\,\Pair[G]$ of
  the monad $\FinSupp$ over $\Pair[G]$ is s special case of a distributive law
  of the monad $\FinSupp$ over the endofunctor $\Forget(G) \times {-}$.
  The latter is exactly a \emph{strength} for the monad $\FinSupp$ with respect
  to the Cartesian product in $\mathbf{Set}$.
  Such monad strengths are known to exist uniquely
  \cite[§3.1]{McDermott_2022}.  They take the form
  \begin{equation*}
    \operatorname\lambda(x, y)\mathpunct. T\bigl(\operatorname\lambda z\mathpunct.(x, z)\bigr)(y)
  \end{equation*}
  for any $\mathbf{Set}$-monad $T$. 

  Proceed by contradiction.  Let $G$ be nontrivial.  Pick some arbitrary
  $a \in \Forget(G)$ with $a \neq e$ and hence $a^2 \neq a$, and some
  $x \in X$.  Now let $p = \bigl(a, \delta_X(x)\bigr)$, and hence
  $p^2 = \bigl(a^2, \delta_X(x)^2\bigr)$.  Then we get, on the one hand,
  \begin{equation*}
    \tau_G(p) = \delta_{\Forget(G) \times X}(a, x)
  \end{equation*}
  and hence
  \begin{equation*}
    \tau_G(p)^2 = \delta_{\Forget(G) \times X}(a, x)^2
  \end{equation*}
  but on the other hand:
  \begin{equation*}
    \tau_G(p^2) = \delta_{\Forget(G) \times X}(a^2, x)^2
  \end{equation*}

  If $\tau_G$ were a group homomorphism, then the latter two had to coincide;
  but that is not the case by Lemma~\ref{lemma:powers}.
\end{proof}

\begin{remark}
  Normalization is generally not injective, since the actual distribution of
  partial prefixes is also ``multiplied out'', and forgotten.  For example,
  \emph{micrometer-per-microsecond} cancels to \emph{meter-per-second}:
  \begin{align*}
    \operatorname{norm}\bigl(\bigl\{ (\delta(\mathrm{\upmu}), \mathrm{m}) \mapsto 1,
    (\delta(\mathrm{\upmu}), \mathrm{s}) \mapsto -1\bigr\}\subzero\bigr) &=
    \operatorname{norm}\bigl(\bigl\{ \eta(\mathrm{m}) \mapsto 1, \eta(\mathrm{s}) \mapsto -1
    \bigr\}\subzero\bigr) = \theta(\mathrm{m})\,\theta(\mathrm{s})^{-1}
  \end{align*}
\end{remark}

\begin{definition}[Normal Equivalence]
  \label{def:norm-equiv}
  Two units are called \emph{normally equivalent}, written $\simeq\subnorm$, iff
  their normalizations coincide:
  \begin{align*}
    u \simeq\subnorm v \iff \operatorname{norm}(u) = \operatorname{norm}(v)
  \end{align*}
\end{definition}

\begin{example}[Precipitation]
  \label{ex:precip}
  Normalization is not part of the traditional notation of units.  However, it
  has the beneficial property that entangled, redundant prefixes and base units
  are cancelled out orthogonally:

  Consider a meteorological unit $p$ for \emph{amount of precipitation}, denoting
  \emph{litre-per-square-meter} ($\mathrm{L}/\mathrm{m}^2$), where a
  \emph{litre} ($\mathrm{L}$) is defined as the third power of a \emph{decimeter}
  ($\mathrm{dm}$), which parses as a simple preunit analogous to $\mathrm{kg}$;
  that is formally:
  \begin{align*}
    p &= \bigl\{ (\delta(\mathrm{d}), \mathrm{m}) \mapsto 3, \eta(\mathrm{m}) \mapsto -2 \bigr\}\subzero
  \end{align*}
  By normalization, a root unit factor of $\delta(\mathrm{m})^2$ is cancelled out:
  \begin{align*}
    \operatorname{norm}(p) &= \bigl( \delta(\mathrm{d})^3, \delta(\mathrm{m}) \bigr)
  \end{align*}
  Thus we find that $p$ is normally equivalent to a \emph{deci-deci-deci-meter},
  but not to a \emph{millimeter};
  \begin{align*}    
    p &\simeq\subnorm \delta\bigl(\bigl(\delta(\mathrm{d})^3, \mathrm{m}\bigr)\bigr)
    &
    p &\not\simeq\subnorm \delta\bigl(\bigl(\delta(\mathrm{m}), \mathrm{m}\bigr)\bigr)
  \end{align*}
\end{example}

In the semantic structures presented so far, the base symbols are free; they
stand only for themselves, operated on exclusively by natural transformations,
and do not carry any attributes for comparison.  As the last example has shown,
the resulting notions of semantic equivalence may be narrower than intended.
The following section introduces one such attribute each for prefixes and units,
and explores the semantic consequences.  Note that any actual assignment of
attribute values is \emph{contingent}; it could well be different in another
possible world, i.e., system of units, whereas all of the preceding reasoning is
logically \emph{necessary}.

\subsection{Specific Attributes}\label{props}

\begin{definition}[Base Prefix Value]
  \label{def:val-base}
  Every base prefix shall be assigned a ratio as its numerical value.
  \begin{align*}
    \operatorname{val}\subbase &: \Pref\subbase \to \RatP
  \end{align*}
\end{definition}

\begin{example}[SI Prefix Values]
  \label{ex:si-val}
  The three families of SI prefixes are defined numerically as negative and
  positive (mostly triple) powers of ten, and positive (tenfold) powers of two,
  respectively:\,\footnote{~It is scientific standard to assign context-free
    numerical values to prefixes; some traditional notations do not follow the
    practice.  For example, consider the popular (ab)use of \emph{kilobyte} for
    $2^{10}$ bytes, which has given rise to the binary family for proper
    distinction.}
    
  \begin{align*}
    \operatorname{val}\subbase^{\mathrm{SI}} &= \left\{
      \begin{array}{l@{{}\mapsto{}}ll@{{}\mapsto{}}lllll}
        \mathrm{d} & 10^{-1}, & \mathrm{c} & 10^{-2}, & \mathrm{m} \mapsto 10^{-3}, &
        \mathrm{\upmu} \mapsto 10^{-6}, & \mathrm{n} \mapsto 10^{-9}, & \dots
        \\
        \mathrm{da} & 10^{+1}, & \mathrm{h} & 10^{+2}, & \mathrm{k} \mapsto 10^{+3}, &
        \mathrm{M} \mapsto 10^{+6}, & \mathrm{G} \mapsto 10^{+9}, & \dots
        \\
        \mathrm{ki} & 2^{+10}, & \mathrm{Mi} & 2^{+20}, & \mathrm{Gi} \mapsto 2^{+30}, & \dots
      \end{array}
    \right\}
  \end{align*}
  etc.
\end{example}

\begin{definition}[Base Unit Dimension]
  Every base unit shall be assigned a (possibly compound) dimension.
  \begin{align*}
    \operatorname{dim}\subbase &: \Unit\subbase \to \Forget(\Dim)
  \end{align*}
\end{definition}

\begin{example}[SI Unit Dimensions]
  The SI base units correspond to the SI base dimensions in the respective order
  presented above, that is
  \begin{align*}
    \operatorname{dim}\subbase^{\mathrm{SI}}(\mathrm{m}) &= \delta(\mathsf{L})
    &
    \operatorname{dim}\subbase^{\mathrm{SI}}(\mathrm{s}) &= \delta(\mathsf{T})
  \end{align*}
  etc., whereas derived SI units may have more complex dimensions:
  \begin{align*}
    \operatorname{dim}\subbase^{\mathrm{SI}}(\mathrm{N}) &= \{ \mathsf{L} \mapsto 1, \mathsf{T} \mapsto -2, \mathsf{M} \mapsto 1 \}\subzero
  \end{align*}
\end{example}

\begin{definition}[Prefix Value]
  \label{def:val}
  Prefix value assignment lifts naturally to all concerned structures:
  \begin{align*}
    \operatorname{val} &: \Pref \to \RatM & \operatorname{val} &= \varepsilon_\RatM \circ \Abel(\operatorname{val}\subbase)
    \\
    \operatorname{pval}\subpre &: \Unit\subpre \to \RatP & \operatorname{pval}\subpre &= {\operatorname{val}} \circ \pi_1
    \\
    \operatorname{pval} &: \Unit \to \RatM & \operatorname{pval} &= \varepsilon_\RatM \circ \Abel(\operatorname{pval}\subpre)
    \\
    \operatorname{pval}\subnorm &: \Unit\subnorm \to \RatM & \operatorname{pval}\subnorm &= {\operatorname{val}} \circ \pi_1
  \end{align*}
\end{definition}

\begin{definition}[Unit Dimension]
  \label{def:unit-dim}
  Dimension assignment lifts naturally to all concerned structures:
  \begin{align*}
    \operatorname{dim}\subroot &: \Unit\subroot \to \Dim & \operatorname{dim}\subroot &= \varepsilon_\Dim \circ \Abel(\operatorname{dim}\subbase)
    \\
    \operatorname{dim}\subpre &: \Unit\subpre \to \Forget(\Dim) & \operatorname{dim}\subpre &= {\operatorname{dim}\subbase} \circ \pi_2
    \\
    \operatorname{dim} &: \Unit \to \Dim & \operatorname{dim} &= \varepsilon_\Dim \circ \Abel(\operatorname{dim}\subpre)
    \\
    \operatorname{dim}\subnorm &: \Unit\subnorm \to \Dim & \operatorname{dim}\subnorm &= {\operatorname{dim}\subroot} \circ \pi_2
  \end{align*}
\end{definition}

\begin{example}[Density]
  Consider a unit of \emph{mass density}, $\mathrm{kg}/\mathrm{cm}^3$, that is formally
  \begin{align*}
    q &= \bigl\{ (\delta(\mathrm{k}), \mathrm{g}) \mapsto 1, (\delta(c), \mathrm{m}) \mapsto -3 \bigr\}\subzero
  \end{align*}
  In the SI interpretation, it follows that
  \begin{align*}
    \operatorname{pval}(q) &= 10^3 \cdot (10^{-2})^{-3} = 10^9
    &
    \operatorname{dim}(q) = \{ \mathsf{L} \mapsto -3, \mathsf{M} \mapsto 1 \}\subzero
  \end{align*}
\end{example}

\begin{definition}[Codimensionality]
  Two (base, pre-, \dots) units are called \emph{codimensional}, written
  $\sim_{\Box}$, with the appropriate initial substituted for $\Box$,
  iff their assigned dimensions coincide:
  \begin{align*}
    u \sim_{\Box} v \iff \operatorname{dim}_{\Box}(u) = \operatorname{dim}_{\Box}(v)
  \end{align*}
\end{definition}

\begin{definition}[Evaluated Unit]
  The \emph{evaluated units} $\Unit\subeval$ are the direct sum of a ratio and a
  root unit:
  \begin{align*}
    \Unit\subeval &= \GPair[\RatM](\Unit\subroot)
  \end{align*}
\end{definition}

\begin{definition}[Evaluation]
  Evaluation separates prefix value and root unit:
  \begin{align*}
    \operatorname{eval} &: \Unit \to \Unit\subeval & \operatorname{eval} &= \langle \operatorname{pval}, \operatorname{root} \rangle = \operatorname{eval}\subnorm \circ \operatorname{norm}
    \\
    \operatorname{eval}\subnorm &: \Unit\subnorm \to \Unit\subeval & \operatorname{eval}\subnorm &= \operatorname{val} \times \operatorname{id}_{\Unit\subroot}
  \end{align*}
\end{definition}

\begin{definition}[Numerical Equivalence]
  \label{def:eval-equiv}
  Two units are called \emph{numerically equivalent}, written $\simeq\subeval$,
  iff their evaluations coincide:
  \begin{align*}
    u \simeq\subeval v \iff \operatorname{eval}(u) = \operatorname{eval}(v)
  \end{align*}
\end{definition}

\begin{theorem}
  \label{thm:norm-hier}
  Normalization equivalence entails numerical equivalence, which entails root
  equivalence, which entails codimensionality:
  \begin{align*}
    u \simeq\subnorm v \implies u \simeq\subeval v \implies u \simeq\subroot \implies u \sim v
  \end{align*}
\end{theorem}

\begin{proof}
  The relations arise from maps by successive composition:
  \begin{enumerate}
  \item The relation $\simeq\subnorm$ is the kernel of the map $\mathrm{norm}$.  
  \item The relation $\simeq\subeval$ is the kernel of the map
    $\mathrm{eval} = (\operatorname{val} \times
    \operatorname{id}_{\Unit\subroot}) \circ \mathrm{norm}$.
  \item The relation $\simeq\subroot$ is the kernel of the map
    $\operatorname{root} = \pi_2 \circ \mathrm{eval}$.  
  \item The relation $\sim$ is the kernel of the map
    $\operatorname{dim} = {\operatorname{dim}\subroot} \circ
    {\operatorname{root}}$.
  \end{enumerate}
  Thus the kernel relations increase with each step.  %\qed
\end{proof}

Numerical equivalence is coarser than normalization equivalence, because it uses
an additional source of information\NEW{: W}hereas the latter depends only on the
universal group structures of prefixes and units, the former takes contingent
value assignments ($\operatorname{val}\subbase$) into account.

\begin{example}[Precipitation revisited]
  Continuing Example~\ref{ex:precip}, we find that, given the SI interpretation
  of prefix values, the precipitation unit $p$ is indeed numerically equivalent
  to the \emph{millimeter}:
  \begin{align*}
    \operatorname{eval}^{\mathrm{SI}}(p) = \bigl(10^{-3}, \delta(\mathrm{m})\bigr) = \operatorname{eval}^{\mathrm{SI}}\Bigl(\delta\bigl(\bigl(\delta(\mathrm{m}), \mathrm{m}\bigr)\bigr)\Bigr) \implies p \simeq\subeval \delta\bigl(\bigl(\delta(\mathrm{m}), \mathrm{m}\bigr)\bigr)
  \end{align*}
\end{example}

So far, three points on the syntax--semantics axis have been discussed, namely
units proper as the formal model of traditional syntax, and normalized and
evaluated units as two steps of semantic interpretation, giving rise to
successively coarser equivalence relations.  Evaluated units are the coarsest
interpretation that is arguably adequate in general practice.
For the sake of critical argument, a further semantic interpretation step
shall be considered.  That one fits the logical framework as well, see
Figure~\ref{fig:structure} below, but is too coarse for sound reasoning in
most cases.
  
\begin{definition}[Abstract Unit]
  The \emph{abstract units} $\Unit\subabs$ are the direct sum of a ratio and a
  dimension:
  \begin{align*}
    \Unit\subabs &= \GPair[\RatM](\Dim)
  \end{align*}
\end{definition}

\begin{definition}[Abstraction]
  Abstraction separates prefix value and dimension assignments:
  \begin{align*}
    \operatorname{abs} &: \Unit \to \Unit\subabs & \operatorname{abs} &= \langle \operatorname{pval}, \operatorname{dim} \rangle = \operatorname{abs}\subnorm \circ \operatorname{norm}
    \\
    \operatorname{abs}\subnorm &: \Unit\subnorm \to \Unit\subabs & \operatorname{abs}\subnorm &= \langle \operatorname{pval}\subnorm, \operatorname{dim}\subnorm \rangle = \operatorname{val} \times \operatorname{dim}\subroot = \operatorname{abs}\subeval \circ \operatorname{eval}\subnorm
    \\
    \operatorname{abs}\subeval &: \Unit\subeval \to \Unit\subabs & \operatorname{abs}\subeval &= \GPair[\RatM](\operatorname{dim}\subroot)
  \end{align*}
\end{definition}

Abstract units arise in the remodeling of theoretical works that assert one
canonical unit per dimension \cite[e.g.]{Allen2004,Kennedy1996}; see also
Section~\ref{re-dim-aargh} below.

\begin{ENEW}
\begin{example}[Pecipitation Abstracted]
  The precipitation unit $p$ abstracts to
  \begin{equation*}
    \operatorname{abs}^{\mathrm{SI}}(p) = \bigl(\operatorname{pval}^{\mathrm{SI}}(p), \operatorname{dim}^{\mathrm{SI}}(p)\bigr) = (10^{-3}, \mathsf{L})
  \end{equation*}
  which means ``one thousandth of the coherent SI unit of length''.
\end{example}
\end{ENEW}

\subsection{The Big Picture}

\begin{figure}
  \begin{equation*}
    \xymatrix{
      &
      \Pref\subbase
      \ar[ld]_{\operatorname{val}\subbase}
      \ar[d]_{\delta}
      &
      \Unit\subpre
      \ar[ld]_{\pi_1}
      \ar[d]_{\delta}
      \ar[r]^{\pi_2}
      &
      \Unit\subbase
      \ar[d]_{\delta}
      \ar[rd]^(.45){\operatorname{dim}\subbase}
      &
      \Dim\subbase
      \ar[d]^{\delta}
      \\
      \RatM
      &
      \Pref
      \ar[l]_-{\operatorname{val}}
      &
      \Unit
      \ar[l]_-{\operatorname{pval}}
      \ar[d]|{\vphantom{qX}\operatorname{norm}}
      \ar[r]^-{\operatorname{root}}
      &
      \Unit\subroot
      \ar[r]^-{\operatorname{dim}\subroot}
      &
      \Dim
      \\
      &&
      \Unit\subnorm
      \ar[lu]^{\pi_1}
      \ar[ru]_(.3){\pi_2}
      \ar[d]|{\operatorname{eval}\subnorm}
      \\
      &&
      \Unit\subeval
      \ar@/^/[lluu]_{\pi_1}
      \ar@/_/[ruu]_{\pi_2}
      \ar[d]|{\operatorname{abs}\subeval}
      \\
      &&
      \Unit\subabs
      \ar@/^/[lluuu]^{\pi_1}
      \ar@/_/[rruuu]_{\pi_2}
    }
  \end{equation*}
  \caption{Overview of Model Structures}
  \label{fig:structure}
\end{figure}

Figure~\ref{fig:structure} depicts all mathematical structures defined in the
preceding sections and their relating operations in a single diagram.  The
diagramm commutes, such that, for every pair of objects, all paths from one to
the other denote the same function, \NEW{namely} \emph{the} model operation of
that type signature.

The depiction abuses notation mildly for the sake of conciseness; the top row
and the rest of the diagram live in categories $\mathbf{Set}$ and $\mathbf{Ab}$,
respectively.  However, the perspective shift is easily explained: The
adjunction of functors $\Abel \dashv \Forget$ ensures that, for any map $f$ that
takes values from the carrier of a group $G$, there exists a unique homomorphic
extension $f'$ to the free abelian group over the domain, i.e., the \NEW{corresponding}
diagrams commute simultaneously in $\mathbf{Set}$ and $\mathbf{Ab}$,
respectively; see Figure~\ref{fig:adj}.

\begin{figure}
  \begin{align*}
    \xymatrix{
      X
      \ar[rd]^{f}
      \ar[dd]_{\delta_X}
      \\
      & \Forget(G)
      \\
      \Forget\bigl(\Abel(X)\bigr)
      \ar@{-->}[ur]_{\Forget(f')}
    }
    &&
    \xymatrix{
      & \Abel\bigl(\Forget(G)\bigr)
      \ar[dd]^{\varepsilon_G}
      \\
      \Abel(X)
      \ar[ru]^{\Abel(f)}
      \ar@{-->}[rd]_{\exists!f'}
      \\
      & G
    }
  \end{align*}
  \caption{Correspondence Between $\mathbf{Set}$ and $\mathbf{Ab}$ by Adjunction}
  \label{fig:adj}
\end{figure}

All horizontal arrows in the second row of Figure~\ref{fig:structure} arise from
this principle (Definitions~\ref{def:val},~\ref{def:unit-dim}).  The syntactic model
$\Unit$ is in the center, its semantics projections extend to the left and
right.  The two independent sources of contingent data are in the top left and
right corner, respectively; all other operations arise from adjoint functors and
their natural transformations.  Extending downward from the center is a sequence
of componentwise semantic interpretations.  Named operations that are just
compositions of others are not shown.

\section{Unit Conversion Relations}
\label{rels}

In the following sections, we deal with ternary relations of a particular
type, namely the carrier of a direct sum of three groups:
\begin{align*}
  \mathrm{Conv} &= \Unit \times \RatM \times \Unit &
  \Forget(\mathrm{Conv}) &= \FinSupp(\Unit\subpre) \times \RatP \times \FinSupp(\Unit\subpre)
\end{align*}
The middle component is called a \emph{conversion ratio}.  We shall write
$\conv[R]urv$ for $(u, r, v) \in R$, a notation that alludes to categorial
diagrams; see Theorem~\ref{thm:conv-cat} below for justification.  As with
diagrams, the subscript $R$ shall be omitted where it is obvious from the
context, and adjacent triples $\conv urv,\conv vsw$ shortened to
$\conv ur{\conv vsw}$.

\subsection{\NEW{Unit Conversions and Closures}}

\begin{definition}[Unit Conversion]
  A relation $C \subseteq \Forget(\mathrm{Conv})$ is called a \emph{(unit)
    conversion}, iff it is dimensionally consistent and functional in its
  arguments of unit type:
  \begin{subequations}
    \begin{align}
      \conv[C]urv &\implies u \sim v \label{eq:conv-dim}
      \\
      \conv[C]urv \land \conv[C]u{r'}v &\implies r = r' \label{eq:conv-fun}
    \end{align}
  \end{subequations}
\end{definition}

\begin{remark}
  \label{rem:conv}
  A judgement $\conv[C]urv$ means ``One $u$ is $r$ $v$s,'' and hence
  introduces an algebraic rewriting rule for quantity values that converts
  (values measured in) $u$ to $v$, since multiplication with the
  conversion factor $r$ is taken as associative:
  \begin{equation}
    \label{eq:conv}
    x\,u = x\,(r\,v) = (xr)\, v
  \end{equation}
\end{remark}

\begin{example}[UK Units]
  \label{ex:uk}
  Standards in the United Kingdom define the regional customary units
  \emph{pound} ($\mathrm{lb}$) and \emph{pint} ($\mathrm{pt}$) with the following
  conversion rules:
  \begin{align*}
    \conv%[\mathrm{SI}^+]
    {\de{\mathrm{lb}}}{453.592\,37}{\de{\mathrm{g}}}
    &&
    \conv%[\mathrm{SI}^+]
    {\de{\mathrm{pt}}}{568.261\,25}{\delta\bigl(\bigl(\delta(\mathrm{c}),\mathrm{m}\bigr)\bigr)^3}
  \end{align*}
\end{example}

\begin{definition}[Conversion Closure]
  Let $C \subseteq \Forget(\mathrm{Conv})$ be a unit conversion.  Its
  \emph{(conversion) closure} is the smallest relation $C^*$ with
  $C \subseteq C^* \subseteq \Forget(\mathrm{Conv})$ such that the following
  axioms hold:
  \begin{subequations}
    \begin{align}
      \conv[C^*]urv \land \conv[C^*]{u'}{r'}{v'} &\implies \conv[C^*]{uu'}{rr'\kern-1pt}{vv'} \label{eq:close-mult}
      \\
      \conv[C^*]urv &\implies \conv[C^*]{u^{-1}}{r^{-1}\kern-1pt}{v^{-1}} \label{eq:close-inv}
      \\
      &\hspace{-2cm} \conv[C^*]u{\operatorname{pval}(u)}{\conv[C^*]{\operatorname{strip}(u)}{\operatorname{pval}(u)^{-1}}{u}} \label{eq:close-dist}
    \end{align}
  \end{subequations}
\end{definition}

\begin{remark}
  The closure of a unit conversion may fail to be a unit conversion itself,
  since contradictory factors can arise from closure axioms,
  violating~\eqref{eq:conv-fun}.  For example, take both $\conv[C]{u}{2}{v}$ and
  $\conv[C]{u^{-1}}{3}{v^{-1}}$, but clearly $2 \neq 3^{-1}$.
\end{remark}

\begin{remark}
  Closure ensures operational \emph{completeness} of the reasoning, namely that
  all potential rewriting rules for composite units implied by rules for their
  constituents are actually available.  For example, if we know how to convert
  from \emph{furlongs} to \emph{meters}, and from \emph{seconds} to
  \emph{fortnights}, then we can deduce how to convert from
  \emph{furlongs-per-fortnight} to \emph{millimeters-per-second}.
\end{remark}

Conversion closures have a rich algebraic structure:

\begin{theorem}
  \label{thm:conv-sub}
  The closure of a unit conversion forms a subgroup of $\mathrm{Conv}$.
\end{theorem}

\begin{proof}
  Axioms \eqref{eq:close-mult} and \eqref{eq:close-inv} state directly that a
  unit conversion closure is closed under the direct sum group operation and
  inversion.  The closure is nonempty by virtue of \eqref{eq:close-dist}; at
  least the triple $\conv{\varnothing\subzero}{1}{\varnothing\subzero}$ is
  always contained.  Any subset of a group with these three properties is a
  subgroup.
\end{proof}

\begin{theorem}
  \label{thm:conv-cat}
  The closure of a unit conversion forms an invertible category
  \NEW{or groupoid}, the categorial generalization \NEW{both of a
    group and} of an equivalence relation, with unit objects and
  conversion factor morphisms:
  \begin{subequations}
    \begin{align}
      &\conv[C^*]u1u \label{eq:close-cat-id}
      \\
      \conv[C^*]urv \conv[C^*]{}sw &\implies \conv[C^*]u{rs}w \label{eq:close-cat-comp}
      \\
      \conv[C^*]urv &\implies \conv[C^*]v{r^{-1}}u \label{eq:close-cat-inv}
    \end{align}
  \end{subequations}
\end{theorem}

\newcommand\infer[3][]{\dfrac{#2}{#3}\hbox to 0pt{\,\text{#1}\hss}}

\begin{shortproof}
  \begin{itemize}
  \item For \eqref{eq:close-cat-id}: By \eqref{eq:close-dist} we have
    $\conv[C^*]{u}{\operatorname{pval}(u)}{\operatorname{strip}(u)}$ and
    $\conv[C^*]{\operatorname{strip}(u)}{\operatorname{pval}(u)^{-1}}{u}$.
    Multiplication via \eqref{eq:close-mult} yields:
    \begin{align*}
      \conv[C^*]{u \operatorname{strip}(u)}{1}{u \operatorname{strip}(u)}
    \end{align*}
    By substitution of $\operatorname{strip}(u)$ for $u$ in \eqref{eq:close-dist},
    we also obtain
    $\conv[C^*]{\operatorname{strip}(u)}{\operatorname{pval}(\operatorname{strip}(u))}{\operatorname{strip}(\operatorname{strip}(u))}$,
    which simplifies via Lemma~\ref{lemma:strip-idemp} to
    $\conv[C^*]{\operatorname{strip}(u)}{1}{\operatorname{strip}(u)}$, and
    inversion via \eqref{eq:close-inv} yields:
    \begin{align*}
      \conv[C^*]{\operatorname{strip}(u)^{-1}}{1}{\operatorname{strip}(u)^{-1}}
    \end{align*}
    Multiplication of both via \eqref{eq:close-mult} concludes:
    \begin{align*}
      \conv[C^*]{u}{1}{u}
    \end{align*}
    
  \item For \eqref{eq:close-cat-comp} assume $\conv[C^*]urv$ and
    $\conv[C^*]vsw$.  Multiplication via \eqref{eq:close-mult} yields:
    \begin{align*}
      \conv[C^*]{uv}{rs}{vw}
    \end{align*}
    By substitution of $v^{-1}$ for $u$ in \eqref{eq:close-cat-id}, we also obtain
    $\conv[C^*]{v^{-1}}{1}{v^{-1}}$.
    Multiplication of both via \eqref{eq:close-mult} concludes:
    \begin{align*}
      \conv[C^*]{u}{rs}{w}
    \end{align*}
    
  \item For \eqref{eq:close-cat-inv} assume $\conv[C^*]urv$.  Inversion via
    \eqref{eq:close-inv} yields:
    \begin{align*}
      \conv[C^*]{u^{-1}}{r^{-1}}{v^{-1}}
    \end{align*}
    By \eqref{eq:close-cat-id}, used both directly and substituting $v$ for $u$,
    we also obtain $\conv[C^*]{u^{-1}}{1}{u^{-1}}$ and
    $\conv[C^*]{v^{-1}}{1}{v^{-1}}$, respectively.  Multiplication of all three
    via \eqref{eq:close-mult} concludes:
    \begin{align*}
      \conv[C^*]{v}{r^{-1}}{u}
    \end{align*}
  \end{itemize}
\end{shortproof}

\begin{remark}
  The closure of a unit conversion is again a conversion iff it is
  \emph{thin} as a category.  Such relations shall take center stage in the next
  section.
\end{remark}

\begin{theorem}
  \label{thm:conv-part}
  Conversions can be decomposed and partitioned by dimension:
  \begin{align*}
    \begin{aligned}
      \Unit_{(d)} &= \{ u \in \Unit \mid \operatorname{dim}(u) = d \} \subseteq \Forget(\Unit)
      \\
      C_{(d)} &= C \cap (\Unit_{(d)} \times \mathbb{Q}_+ \times \Unit_{(d)})
    \end{aligned}
    &&
    C &= \bigcup_{\mathclap{d \in \Dim}} C_{(d)}
  \end{align*}
\end{theorem}

\begin{shortproof}%[Theorem~\ref{thm:conv-part}]
  Follows directly from axiom \eqref{eq:conv-dim}.  %\qed
\end{shortproof}

\begin{remark}
  A non-converting theory of units, such as \cite{Kennedy1996}, is
  characterized by being partitioned into \emph{trivial} subgroups.  This
  demonstrates that the theory being presented here is a complementary extension
  to previous work.
\end{remark}

\begin{remark}
  Closure cannot be performed on the partitions individually, mostly because of
  axiom \eqref{eq:close-mult} that allows for multiplication of units with
  orthogonal dimensions.  Thus it is generally the case that:
  \begin{align*}
    C^* &\neq \bigcup_{\mathclap{d \in \Dim}} {C_{(d)}}^*  
  \end{align*}
\end{remark}

\begin{definition}[Unit Convertibility]
  Two units are called \emph{convertible}, with respect to a unit conversion
  $C$, written $\propto_C$, iff they are related by some factor:
  \begin{align*}
    u \propto_C v \iff \conv[C]u{\exists r}v
  \end{align*}
\end{definition}

\begin{definition}[Unit Coherence]
  \label{def:unit-coherent}
  Two units are called \emph{coherent}, with respect to a unit conversion $C$,
  written $\cong_C$, iff they are related by the factor one:
  \begin{align*}
    u \cong_C v \iff \conv[C]u1v
  \end{align*}
\end{definition}

\begin{theorem}
  \label{thm:conv-equiv-hier}
  Coherence entails convertibility, which entails codimensionality:
  \begin{align*}
    u \cong_C v \implies u \propto_C v \implies u \sim v
  \end{align*}
\end{theorem}

\begin{shortproof}%[Theorem~\ref{thm:conv-equiv-hier}]
  Follows directly from the definitions and axiom \eqref{eq:conv-dim}.  %\qed
\end{shortproof}

\begin{definition}[Conversion Coherence]
  By extension, a unit conversion $C$ is called \emph{coherent} iff all
  units \NEW{that are convertible with respect to $C$} are coherent:
  \begin{align*}
    u \cong_C v \iff u \propto_C v
  \end{align*}
\end{definition}

\begin{ENOW}
  \begin{remark}
    \label{rem:si-coherent}
    The notion of coherence given here is a formal remodeling of the one used by
    the SI; recall Example~\ref{ex:root-si}.  We take the semantic
    characterization, namely a conversion factor of one, as the actual
    definition.  The pseudo-syntactic characterization shall be recovered later
    as a result for \emph{regular} conversions, that are algebraically complete
    but otherwise sparse; see Theorem~\ref{thm:reg-eval}.
  \end{remark}
\end{ENOW}

\begin{example}[SI Unit Coherence]
  The conversion relation of the seven canonical base units of the SI is
  trivially coherent \NEW{in the more general sense}, since they are pairwise
  inco\NEW{n}vertible.  In addition, the SI recognizes 22 derived units with
  coherent conversion rules, \NEW{see Section~\ref{si-co}}.  By contrast, many
  traditional units, such as the \emph{hour}~($\mathrm{h}$), are convertible but
  not coherently so: $\conv{\mathrm{h}}{3600}{\mathrm{s}}$.
\end{example}

\subsection{\NEW{Defining Unit Conversions and Rewriting}}

Interesting subclasses of conversions arise as the closures of syntactically
restricted generators.

\begin{definition}[Defining Conversion]
  A unit conversion $C$ is called \emph{defining} iff it is basic and
  functional in its first component:
  \begin{subequations}
    \begin{align}
      \conv[C]urv &\implies \operatorname\exists u_0\mathpunct. u = \de{u_0}
      \\
      \conv[C]urv \land \conv[C]u{r'}{v'} &\implies v = v'
    \end{align}
  \end{subequations}
\end{definition}

\begin{definition}[Definition Expansion]
  Every defining unit conversion $C$ gives rise to a totalized \emph{expansion}
  function:
  \begin{align*}
    \operatorname{xpd}(C) &: \Unit\subbase \to \Pair[\RatM](\Unit) & \operatorname{xpd}(C)(u_0) &=
    \begin{cases}
      (r, v) & \text{if~} \conv[C]{\de{u_0}}rv
      \\
      \eta(\de{u_0}) & \text{if no match}
    \end{cases}
  \end{align*}
  This in turn gives rise to a mapping of base to evaluated units, and ultimately
  to an iterable rewriting operation on evaluated units:
  \begin{align*}
    \operatorname{rwr}\subbase(C) &: \Unit\subbase \to \Unit\subeval & \operatorname{rwr}\subbase(C) &= \mu_{\RatM,\Unit\subroot} \circ \Pair[\RatM](\operatorname{eval}) \circ \operatorname{xpd}(C)
    \\
    \operatorname{rwr}\subeval(C) &: \Unit\subeval \to \Unit\subeval & \operatorname{rwr}\subeval(C) &= \xi_{\RatM,\Unit\subbase} \circ \Pair[\RatM]\FinSupp\bigl(\operatorname{rwr}\subbase(C)\bigr)
  \end{align*}
\end{definition}

\begin{definition}[Dependency Order]
  Let $C$ be a defining unit conversion.  The relation
  $(>_C) \subseteq \Forget(\Unit\subbase)^2$ is the smallest transitive relation such
  that:
  \begin{align*}
    \de{u_0} \propto_C v \land v_0 \in \operatorname{supp}\bigl(\operatorname{root}(v)\bigr) \implies u_0 >_C v_0
  \end{align*}
  We say that, with respect to $C$, $u_0$ depends on $v_0$.
\end{definition}

\begin{definition}[Well-Defining Conversion]
  \label{def:well-def}
  A defining unit conversion $C$ is called \emph{well-defining} iff its
  dependency order $>_C$ is well-founded.  Since $\Unit\subbase$ is finite, this
  is already the case if $>_C$ is antireflexive.
\end{definition}

\begin{theorem}
  \label{thm:rew-fix}
  Let $C$ be a well-defining conversion.  The iteration of
  $\operatorname{rwr}\subeval(C)$ has a fixed point, which is reached after a
  number $N_C$ of steps that is bounded by the number of distinct base units,
  and independent of the input unit:
  \begin{align*}
    \lim_{n \to \infty}{\operatorname{rwr}\subeval(C)}^n &= {\operatorname{rwr}\subeval(C)}^{N_C \leq \lvert \Unit\subbase \rvert}
  \end{align*}
  The proof requires some auxiliary machinery, and is given in
  Section~\ref{rwr-fix} below.
\end{theorem}

\begin{example}[SI Derived Units]
  \label{ex:si-well}%
  Traditional systems of units, including the SI, follow a well-defining
  approach: Starting from an irreducible set of base units, additional
  ``derived'' base units are added in a stratified way, by defining them as
  convertible to (expressions over) preexisting ones.  The SI ontology is rather
  complex; the details are discussed in Section~\ref{si-co} below.
\end{example}

\subsubsection{Rewriting Termination}
\label{rwr-fix}

\begin{definition}[Unit Depth]
  Let $C$ be a well-defining conversion.  The \emph{domain} of $C$ is the set of
  base units occurring in its first component.
  \begin{align*}
    \operatorname{dom}(C) &= \bigl\{ u_0 \bigm| \conv[C]{\de{u_0}}{\exists r}{\exists v} \bigr\}
  \end{align*}
  The \emph{depth} of a (base/root/evaluated) unit is defined by well-founded
  recursion as follows:
  \begin{align*}
    d\subbase^C &: \Unit\subbase \to \mathbb{N} & d\subbase^C(u_0) &=
    \begin{cases}
      1 + \operatorname{sup} \{ d\subbase^C(v_0) \mid u_0 >_C v_0 \} & u_0 \in \operatorname{dom}(C)
      \\
      0 & \text{otherwise}
    \end{cases}
    \\
    d\subroot^C &: \Unit\subroot \to \mathbb{N} & d\subroot^C(u) &= \operatorname{sup} \{ d\subbase^C(u_0) \mid u_0 \in \operatorname{supp}(u) \}
    \\
    d\subeval^C &: \Unit\subeval \to \mathbb{N} & d\subeval^C &= d\subroot^C \circ \pi_2
  \end{align*}
  Note that $\operatorname{sup} \varnothing = 0$, such that a base unit $u_0$
  with $\conv[C]{\de{u_0}}{r}{\varnothing\subzero}$ (for example, the
  \emph{dozen} with $r = 12$) has depth $1$, not $0$.
\end{definition}

\begin{lemma}
  Base unit depth is monotonic with respect to dependency.
  \begin{align*}
    u_0 >_C v_0 \implies d\subbase^C(u_0) > d\subbase^C(v_0)
  \end{align*}
\end{lemma}

\begin{shortproof}
  Follows directly from the definition.  %\qed
\end{shortproof}

\begin{lemma}
  \label{lemma:depth-bound}
  Evaluated unit depth is linearly bounded.
  \begin{align*}
    d\subeval^C(u) \leq \lvert \Unit\subbase \rvert
  \end{align*}
\end{lemma}

\begin{shortproof}
  Follows from $>_C$ being transitive and antireflexive, hence cycle-free.  %\qed
\end{shortproof}

\begin{lemma}
  \label{lemma:depth-term}
  Evaluated unit depth is a termination function for rewriting.
  \begin{align*}
    d\subeval^C(u) = 0 &\implies \operatorname{rwr}\subeval(C)(u) = u
    \\
    d\subeval^C(u) > 0 &\implies d\subeval^C(u) > d\subeval^C\bigl(\operatorname{rwr}\subeval(C)(u)\bigr)
  \end{align*}  
\end{lemma}

\begin{shortproof}
  \begin{enumerate}
  \item Assume that $d\subeval^C(u) = 0$.  Expand
    $\pi_2(u) = \delta(u_1)^{z_1} \cdots \delta(u_n)^{z_n}$ by
    Lemma~\ref{lemma:powers}.  A longish but straightforward calculation, mostly
    with natural transformations, yields $\operatorname{rwr}\subeval(C)(u) = u$.
  \item Assume that $n = d\subeval(u) > 0$.  Then
    $M^C_u = \{ v_0 \in \operatorname{supp}(\pi_2(u)) \mid d\subbase^C(v_0) = n - 1
    \}$ are the base units occuring in $u$ that have maximal depth.  Each of
    these is replaced in $\operatorname{rwr}\subeval(C)(u)$ by base units of
    strictly lesser depth.  Thus
    $d\subeval^C(\operatorname{rwr}\subeval(C)(u)) \leq n - 1$.
  \end{enumerate}
  %\qed
\end{shortproof}

\begin{shortproof}(Theorem~\ref{thm:rew-fix}) \\
  The depth function $d\subeval^C$ provides a bound on the number of iterations
  of $\operatorname{rwr}\subeval(C)$ required to fix a particular element.
  Namely, from Lemma~\ref{lemma:depth-term} it follows that:
  \begin{align*}
    n \geq d\subeval^C(u) \implies {h_C}^n(u) = u
  \end{align*}
  By maxing over all possible base unit depths we obtain a global bound.
  \begin{align*}
    N_C &= \operatorname{max} \{ d\subbase^C(u_0) \mid u_0 \in \Unit\subbase \}
    &
    d\subeval^C(u) \leq N_C
  \end{align*}
  It follows that
  \begin{align*}
    n \geq N_C \implies \operatorname{rwr}\subeval(C)^n(u) &= u
  \end{align*}
  and, by Lemma~\ref{lemma:depth-bound}
  \begin{align*}
    n \geq \lvert \Unit\subbase \rvert \implies {h_C}^n(u) &= u
  \end{align*}
  %\qed
\end{shortproof}
  
\subsubsection{Extended Example: SI Derived Units}
\label{si-co}
     
This section gives an overview of the named coherent SI units, namely the seven
irreducible base units and 22 ``derived units with special names''.  In terms of
the model presented here, of course all 29 are base units, and being irreducible is
synonymous with being $>_C$-minimal.

Table~\ref{tab:si} lists the standardized names and reductionistic definitions
of all ``derived coherent SI units with special names'' \cite{sib9}.  The
definitions are usually given as equations, which is quite problematic: Two
terms being equated, without the possibility to substitute one for the other,
violates the principle of the indiscernibility of identicals.  For easy
reference, all tabulated definitions have been normalized, and then expanded in
terms of exponents of the relevant base prefix ($\mathrm{k}$) and the seven
irreducible base units.

In the model presented here, the table forms a well-defining coherent
conversion.  Figure~\ref{fig:si-hasse} depicts the corresponding dependency
relation $>_C$ in terms of a Hasse diagram.  Dotted lines indicate pathological
cases to be discussed in Section~\ref{cases-rev-cancel}.  \NEW{Compare
also \cite{nist-sp1247}.}

\begin{table}
  \caption{SI derived units with special names \cite{ISO80000-1}}
  \label{tab:si}
  \vspace{-\abovedisplayskip}
  \begin{equation*}
    \renewcommand\arraystretch{1.1}
    \newcommand\0{\hphantom{+}0}
    \arraycolsep 1ex
    \begin{array}{cccccccccc}
      \toprule
      \textbf{Symbol} & \textbf{Definition} & \multicolumn{8}{l}{\textbf{Normalization Exponents}}
      \\
      && \operatorname{pref} & \multicolumn{7}{l}{\operatorname{root}}
      \\
      && \mathrm{k} & \mathrm{m} & \mathrm{g} & \mathrm{s} & \mathrm{A} & \mathrm{K} & \mathrm{mol} & \mathrm{cd}
      \\ \midrule
      \mathrm{rad} & \mathrm{m}/\mathrm{m}\;{}^* & \0 & \0 & \0 & \0 & \0 & \0 & \0 & \0
      \\
      \mathrm{sr} & \mathrm{m}^2/\mathrm{m}^2\;{}^* & \0 & \0 & \0 & \0 & \0 & \0 & \0 & \0
      \\
      \mathrm{Hz} & \mathrm{s}^{-1} & \0 & \0 & \0 & -1 & \0 & \0 & \0 & \0
      \\
      \mathrm{N} & \mathrm{kg}{\cdot}\mathrm{m}/\mathrm{s}^2 & +1 & +1 & +1 & -2 & \0 & \0 & \0 & \0
      \\
      \mathrm{Pa} & \mathrm{N}/\mathrm{m}^2 & +1 & -1 & +1 & -2 & \0 & \0 & \0 & \0
      \\
      \mathrm{J} & \mathrm{N}{\cdot}\mathrm{m} & +1 & +2 & +1 & -2 & \0 & \0 & \0 & \0
      \\
      \mathrm{W} & \mathrm{J}/\mathrm{s} & +1 & +2 & +1 & -3 & \0 & \0 & \0 & \0
      \\
      \mathrm{C} & \mathrm{A}{\cdot}\mathrm{s} & \0 & \0 & \0 & +1 & +1 & \0 & \0 & \0
      \\
      \mathrm{V} & \mathrm{W}/\mathrm{A} & +1 & +2 & +1 & -3 & -1 & \0 & \0 & \0
      \\
      \mathrm{F} & \mathrm{C}/\mathrm{V} & -1 & -2 & -1 & +4 & +2 & \0 & \0 & \0
      \\
      \mathrm{\Omega} & \mathrm{V}/\mathrm{A} & +1 & +2 & +1 & -3 & -2 & \0 & \0 & \0
      \\
      \mathrm{S} & \mathrm{\Omega}^{-1} & -1 & -2 & -1 & +3 & +2 & \0 & \0 & \0
      \\
      \mathrm{Wb} & \mathrm{V}{\cdot}\mathrm{s} & +1 & +2 & +1 & -2 & -1 & \0 & \0 & \0
      \\
      \mathrm{T} & \mathrm{Wb}/\mathrm{m}^2 & +1 & \0 & +1 & -2 & -1 & \0 & \0 & \0
      \\
      \mathrm{H} & \mathrm{Wb}/\mathrm{A} & +1 & +2 & +1 & -2 & -2 & \0 & \0 & \0
      \\
      \text{\textdegree}\mathrm{C} & \mathrm{K} & \0 & \0 & \0 & \0 & \0 & +1 & \0 & \0
      \\
      \mathrm{lm} & \mathrm{cd}{\cdot}\mathrm{sr} & \0 & \0 & \0 & \0 & \0 & \0 & \0 & +1
      \\
      \mathrm{lx} & \mathrm{lm}/\mathrm{m}^2 & \0 & -2 & \0 & \0 & \0 & \0 & \0 & +1
      \\
      \mathrm{Bq} & \mathrm{s}^{-1} & \0 & \0 & \0 & -1 & \0 & \0 & \0 & \0
      \\
      \mathrm{Gy} & \mathrm{J}/\mathrm{kg} & \0 & +2 & \0 & -2 & \0 & \0 & \0 & \0
      \\
      \mathrm{Sv} & \mathrm{J}/\mathrm{kg} & \0 & +2 & \0 & -2 & \0 & \0 & \0 & \0
      \\
      \mathrm{kat} & \mathrm{mol}/\mathrm{s} & \0 & \0 & \0 & -1 & \0 & \0 & +1 & \0
      \\ \bottomrule
    \end{array}
  \end{equation*}
\end{table}

\begin{figure}
  \centering
  \newcommand\unode[2]{\xunode{#1}{#2}{#1}}
  \newcommand\xunode[3]{\pgfnodebox{#1}[virtual]{#2}{$\mathstrut\mathrm{#3}$}{1pt}{1pt}}
  \newcommand\udep[2]{\pgfnodeconnline{#1}{#2}}
  \newcommand\utdep[2]{
    \begin{pgfscope}
      \color{black!33}
      \pgfsetlinewidth{0.4pt}
      \pgfsetdash{{1mm}{1mm}}{0mm}
      \udep{#1}{#2}
    \end{pgfscope}
  }
  \newcommand\ucdep[2]{
    \begin{pgfscope}
      \color{black!33}
      \pgfsetlinewidth{0.6pt}
      \pgfsetdash{{0.6pt}{1.2pt}}{0mm}
      \udep{#1}{#2}
    \end{pgfscope}
  }
  \begin{pgfpicture}{-1.8cm}{0cm}{9cm}{7.2cm}
    \pgfsetxvec{\pgfpoint{1.8cm}{0cm}}
    \pgfsetyvec{\pgfpoint{0cm}{1.2cm}}
    \pgfnodesetsepstart{1pt}
    \pgfnodesetsepend{1pt}
    \pgfsetlinewidth{0.6pt}
    \unode{m}{\pgfxy(0,0)}
    \unode{g}{\pgfxy(1,0)}
    \unode{s}{\pgfxy(2,0)}
    \unode{A}{\pgfxy(3,0)}
    \unode{K}{\pgfxy(4,0)}
    \unode{mol}{\pgfxy(5,0)}
    \unode{cd}{\pgfxy(-1,0)}
    \unode{rad}{\pgfxy(0.33,1)} \ucdep{rad}{m}
    \unode{sr}{\pgfxy(-0.5,1)} \ucdep{sr}{m}
    \unode{Hz}{\pgfxy(1.5,1.5)} \udep{Hz}{s}
    \unode{Bq}{\pgfxy(2.17,1.5)} \udep{Bq}{s}
    \unode{N}{\pgfxy(1,1)} \udep{N}{g} \udep{N}{m} \udep{N}{s}
    \unode{C}{\pgfxy(3,2.5)} \udep{C}{s} \udep{C}{A}
    \xunode{oC}{\pgfxy(4,1)}{\text{\textdegree}C} \udep{oC}{K}
    \unode{kat}{\pgfxy(5,2.5)} \udep{kat}{s} \udep{kat}{mol}
    \unode{lm}{\pgfxy(-1,2)} \udep{lm}{cd} \udep{lm}{sr}
    \unode{Pa}{\pgfxy(0,2)} \udep{Pa}{N} \utdep{Pa}{m}
    \unode{J}{\pgfxy(1,2)} \udep{J}{N} \utdep{J}{m}
    \unode{W}{\pgfxy(1.67,3)} \udep{W}{J} \utdep{W}{s}
    \unode{V}{\pgfxy(1.67,4)} \udep{V}{W} \udep{V}{A}
    \unode{F}{\pgfxy(3,5)} \udep{F}{V} \udep{F}{C}
    \xunode{O}{\pgfxy(2.33,5)}{\Omega} \udep{O}{V} \utdep{O}{A}
    \unode{S}{\pgfxy(2.33,6)} \udep{S}{O}
    \unode{Wb}{\pgfxy(1,5)} \udep{Wb}{V} \utdep{Wb}{s}
    \unode{T}{\pgfxy(0.67,6)} \udep{T}{Wb} \utdep{T}{m}
    \unode{H}{\pgfxy(1.67,6)} \udep{H}{Wb} \utdep{H}{A}
    \unode{Gy}{\pgfxy(0.67,3)} \udep{Gy}{J} \utdep{Gy}{g}
    \unode{Sv}{\pgfxy(1.17,3)} \udep{Sv}{J} \utdep{Sv}{g}
    \unode{lx}{\pgfxy(-0.5,3)} \udep{lx}{lm} \udep{lx}{m}
  \end{pgfpicture}
  \caption{Hasse diagram $(>_{\mathrm{SI}})$ of coherent SI base units.
    \NEW{Light dashed lines indicate dependencies that are explicit in the
      definition, but omitted in the transitive reduct.  Dotted lines indicate
      pathological cases.}}
  \label{fig:si-hasse}
\end{figure}

\clearpage
\subsection{The Conversion Hierarchy}
\label{hier}

\begin{definition}[Conversion Hierarchy]
  A conversion is called \dots
  \begin{enumerate}
  \item \emph{consistent} iff its closure is again a conversion;
  \item \emph{closed} iff it is its own closure;
  \item \emph{finitely generated} iff it is the closure of a finite conversion;
  \item \emph{defined} iff it is the closure of a defining conversion;
  \item \emph{well-defined} iff it is the closure of a well-defining conversion;
  \item \emph{regular} iff it is the closure of an empty conversion.
  \end{enumerate}
\end{definition}

\begin{theorem}
  \label{thm:hier}
  Each property in the conversion hierarchy entails the preceding.
\end{theorem}

\begin{shortproof}%[Theorem~\ref{thm:hier}]
  %The hierarchical entailments are all straightforward.
  \begin{enumerate}
  \item \emph{Closed entails consistent.} -- If $C = C^*$, and $C$ is a
    conversion, then evidently so is $C^*$.
  \item \emph{Finitely generated entails closed.} -- The closure operator is
    idempotent.  Thus, if $C = B^*$, where $B$ is irrelevantly finite, then also
    $B^*=(B^*)^*=C^*$.
  \item \emph{Defined entails finitely generated.} -- If $C = B^*$ and $B$ is
    defining, then the cardinality of $B$ is bounded by the cardinality of
    $\Unit\subbase$, which is finite.
  \item \emph{Well-defined entails defined.} -- If $C = B^*$ and $B$ is
    well-defining, then $B$ is also defining.
  \item \emph{Regular entails well-defined.} -- If $C = B^*$ and
    $B = \varnothing$, then $B$ is vacuously well-defining.
  \end{enumerate}
\end{shortproof}

\begin{theorem}
  \label{thm:cons}
    Consistency is non-local; namely the following three statements are equivalent in
    the closure of a conversion $C$:
    \begin{enumerate} \def\theenumi{\alph{enumi}}
    \item contradictory factors exist for some pair of units;
    \item contradictory factors exist for all convertible pairs of units;
    \item $\conv[C^*]{\varnothing\subzero}{\neq 1}{\varnothing\subzero}$.
    \end{enumerate}
\end{theorem}

\begin{proof}%[Theorem~\ref{thm:cons}]
  By circular implication.  Assume (a), i.e., both
  $\conv[C^*]{u}{r_1}{v}$ and $\conv[C^*]{u}{r_2}{v}$ with
  $r_1 \neq r_2$.  Multiplication via \eqref{eq:close-mult} of the
  former with the inverse by \eqref{eq:close-inv} of the latter
  yields:
  \begin{align*}
    \conv[C^*]{\varnothing\subzero}{r_1r_2^{-1}}{\varnothing\subzero}
  \end{align*}
  with $s = r_1r_2^{-1} \neq 1$, which is (c).  Now let
  $u' \propto_{C^*} v'$ be any convertible pair, i.e.
  \begin{align*}
    \conv[C^*]{u'}{\exists t_1}{v'}
  \end{align*}
  Multiplication via \eqref{eq:close-mult} with the preceding yields:
  \begin{align*}
    \conv[C^*]{u'}{st_1}{v'}
  \end{align*}
  with $st_1 \neq t_1$, which is (b).  This in turn implies (a) by nonemptiness.
\end{proof}

\begin{theorem}
  \label{thm:conv-equiv}
  For closed conversions, convertibility and coherence are group congruence
  relations.
\end{theorem}

\begin{proof}%[Theorem~\ref{thm:conv-equiv}]
  That $\propto_C$ and $\cong_C$ are equivalences follows directly from
  Theorem~\ref{thm:conv-cat}, by forgetting the morphisms.  Any equivalence that
  is also a subgroup of the direct sum (Theorem~\ref{thm:conv-sub}) is a
  congruence.  %\qed
\end{proof}

\begin{theorem}
  \label{thm:conv-sdc}
  For closed conversions, the resulting category is strict dagger compact.
\end{theorem}

\begin{proof}%[Theorem~\ref{thm:conv-sdc}]
  A dagger compact category combines several structures.  Firstly, it is
  symmetric monoidal, i.e., there is a bifunctor $\otimes$ and an object $I$
  such that:
  \begin{align*}
    A \otimes (B \otimes C) &\cong (A \otimes B) \otimes C & I \otimes A \cong A &\cong A \otimes I & A \otimes B &\cong B \otimes A
  \end{align*}
  In the general case, all isomorphisms $\cong$ are named natural
  transformations subject to additional coherence conditions, but in the strict
  case they are identities, and hence no extra laws are required.  Here, the
  abelian group structures of $\Unit, \RatM$ provide the operations:
  \begin{align*}
    u \otimes v &= uv & r \otimes s &= rs & I &= \varnothing\subzero
  \end{align*}

  Secondly, there is a contravariant involutive functor $\dag$ that is
  compatible with the above:
  \begin{align*}
    f : A \to B &\implies f^\dag : B \to A & (f^\dag)^\dag &= f & (f \otimes g)^\dag &= f^\dag \otimes g^\dag
  \end{align*}
  Here, this is given by the inverse in $\RatM$:
  \begin{align*}
    r^\dag &= r^{-1}
  \end{align*}

  Lastly, the category is compact closed, i.e., endowed with a dual $A^\star$ of each
  object $A$, and natural transformations (not necessarily isomorphisms):
  \begin{align*}
    A \otimes A^\star \Rightarrow I \Rightarrow A^\star \otimes A
  \end{align*}
  In the general case, these are again subject to more coherence conditions, but
  in the strict case they are simply identities.  Here, the dual is given by the
  inverse in $\Unit$:
  \begin{align*}
    u^\star &= u^{-1}
  \end{align*}
  %\qed
\end{proof}

\begin{theorem}
  \label{thm:closed-eval}
  In a closed conversion $C$, any two units are convertible if they are root
  equivalent, and coherent if they are numerically equivalent.
  \begin{align*}
    u \simeq\subroot v &\implies \conv[C]u{\operatorname{pval}(u)\operatorname{pval}(v)^{-1}}v \implies u \propto_C v
    &
    u \simeq\subeval v &\implies u \cong_C v
  \end{align*}
\end{theorem}

\begin{proof}%[Theorem~\ref{thm:closed-eval}]
  Let $C$ be a closed conversion.  Assume $u \simeq\subroot v$.  Unfolding the
  definition of $(\simeq\subroot)$ and applying $\operatorname{unroot}$ to
  both sides yields:
  \begin{equation*}
    \operatorname{strip}(u) = \operatorname{strip}(v)      
  \end{equation*}
  The two halves of \eqref{eq:close-dist} for $u$ and $v$, respectively, yield:
  \begin{align*}
    \conv[C]u{\operatorname{pval}(u)}{\operatorname{strip}(u)} && \conv[C]{\operatorname{strip}(v)}{\operatorname{pval}(v)^{-1}}v
  \end{align*}
  Transitivity via \eqref{eq:close-cat-comp} concludes
  \begin{equation*}
    \conv[C]u{\operatorname{pval}(u)\operatorname{pval}(v)^{-1}}v
  \end{equation*}
  which also proves convertibility, $u \propto_C v$.
  
  Now assume $u \simeq\subeval v$.  From this follows all of the above, and additionally:
  \begin{equation*}
    \operatorname{pval}(u) = \operatorname{pval}(v)
  \end{equation*}
  Hence the conversion ratio simplifies
  \begin{equation*}
    \conv[C]u1v
  \end{equation*}
  which proves coherence, $u \cong_C v$.
\end{proof}

\begin{theorem}
  \label{thm:reg-eval}
  In a regular conversion $C$, in addition to Theorem~\ref{thm:closed-eval}, any
  two units are convertible if \emph{and only if} they are root equivalent, and
  coherent if \emph{and only if} they are numerically equivalent:
  \begin{align*}
    u \simeq\subroot v &\iff u \propto_C v
    &
    u \simeq\subeval v &\iff u \cong_C v
  \end{align*}
\end{theorem}

\begin{proof}
  Since any regular conversion is closed, it suffices to prove the converses of
  Theorem~\ref{thm:closed-eval}.
  \begin{enumerate}
  \item All closure axioms preserve root equivalence: If their preconditions are
    instantiated with root equivalent unit pairs, then so are their conclusions.
    It follows that a regular conversion can only relate root equivalent unit
    pairs:
    \begin{equation*}
      u \propto_C v \implies u \simeq\subroot v
    \end{equation*}
  \item Assume
    $u \cong_C v$.  Then we have the above, and additionally
    \begin{equation*}
      \conv[C]u1v
    \end{equation*}
    This, together with Theorem~\ref{thm:closed-eval} and \eqref{eq:conv-fun}, yields:
    \begin{equation*}
      \operatorname{pval}(u) \operatorname{pval}(v)^{-1} = 1
    \end{equation*}
    Thus we find both
    \begin{align*}
      \operatorname{pval}(u) &= \operatorname{pval}(v)
      &
      \operatorname{root}(u) &= \operatorname{root}(v)
    \end{align*}
    and combine components
    \begin{equation*}
      \operatorname{eval}(u) = \operatorname{eval}(v)
    \end{equation*}
    which proves numerically equivalence, $u \simeq\subeval v$.
  \end{enumerate}
\end{proof}

Most of the entailments of Theorem~\ref{thm:hier} are generally proper, but one is not:

\begin{theorem}
  \label{thm:closed-fg}
  All closed conversions are finitely generated.
\end{theorem}

\begin{proof}%[Theorem~\ref{thm:closed-fg}, Outline]
  Let $C$ be a closed conversion.  The component group
  $\Unit = \Abel(\Unit\subpre)$ is not finitely generated as a group, so neither
  is $\mathrm{Conv}$; thus the result is not immediate.

  Since closure is monotonic, any closed conversion has the corresponding
  regular conversion as a (normal) subgroup.  Thus it suffices to show that the
  remaining quotient group is finitely generated.

  By Theorem~\ref{thm:reg-eval} we have:
  \begin{align*}
    u \simeq\subroot u' \land v \simeq\subroot v' \implies \bigl( u \propto_C v \iff u' \propto_C v' \bigr)
  \end{align*}
  
  Pick a representative pair from each equivalence class, the obvious candidate
  being the stripped one.  Let $B$ be the relation such that:
  \begin{align*}
    \conv[C]urv \iff \conv[B]{\operatorname{strip}(u)}{\operatorname{pval}(u)^{-1}\operatorname{pval}(v)}{\operatorname{strip}(v)}
  \end{align*}
  By reasoning in analogy to Theorem~\ref{thm:closed-eval}, we can show that $B$
  is a conversion and a subgroup of $\mathrm{Conv}$, and that $B^* = C$.

  Thus it suffices to show that $B$ is finitely generated as a group.  Since the
  map $\operatorname{root}$ is bijective between stripped and root units, $B$ is
  isomorphic to a subgroup of the binary direct sum
  $\Unit\subroot \times \Unit\subroot$, namely by forgetting the conversion
  factor.  But $\Unit\subroot \times \Unit\subroot$ is a finitely generated free
  abelian group\,---in the wide sense, namely naturally isomorphic to
  $\Abel(\Unit\subbase \times \Unit\subbase)$---\,whose subgroups are all free and
  finitely generated.  %\qed
\end{proof}

\begin{remark}
  Being finitely generated as a closed conversion in this sense is entailed by
  being finitely generated as an abelian group in the sense of
  Theorem~\ref{thm:conv-sub}, but \emph{not} vice versa; the former could be
  ``larger'' due to axiom \eqref{eq:close-dist}, which adds infinitely many elements.
\end{remark}

\begin{theorem}
  \label{thm:well-cons}
  Every well-defining conversion $C$ is consistent, and hence gives rise to a
  well-defined closure $C^*$.
\end{theorem}

\begin{proof}%[Theorem~\ref{thm:well-cons}, Outline]
  By well-founded induction.  Let $C$ be a well-defining conversion.  Its
  finitely many element triples can be sorted w.r.t.\ $>_C$ in ascending
  topological order of the left hand sides.  Now consider an incremental
  construction of $C^*$, starting from the empty set, and alternatingly adding
  the next element of $C$ and applying the closure operator.  It can be shown
  that no such step introduces contradictory conversion factors. 
\end{proof}

\NEW{\subsection{Well-Defining Unit Conversions and Algorithmic Considerations}}

\begin{remark}
  It follows \NEW{from Theorem~\ref{thm:well-cons}} that a defining
  but non-well-defining conversion can only be inconsistent due to a
  circular dependency.
\end{remark}  

\begin{remark}
  The convertibility problem for closed conversions can encode the \emph{word
    problem} for quotients of unit groups, thus there is potential danger of
  undecidability.  We conjecture that, by Theorem~\ref{thm:closed-fg}, closed
  conversions are \emph{residually finite}, such that a known decision algorithm
  \cite{Robinson1996} could be used in principle.
\end{remark}

For well-defined conversions, a more direct and efficient algorithm exists.

\begin{definition}[Exhaustive Rewriting]
  Given a well-defining conversion $C$, units can be evaluated and exhaustively
  rewritten:
  \begin{align*}
    \operatorname{rwr}^*(C) &: \Unit \to \Unit\subeval &
    \operatorname{rwr}^*(C) &= {\operatorname{rwr}\subeval(C)^{N_C}} \circ {\operatorname{eval}}    
  \end{align*}
  Each step requires a linearly bounded number of group operations.
  The kernel of this map can be upgraded to a ternary
  relation $C^\sharp \subseteq \Forget(\mathrm{Conv})$:
  \begin{align*}
    C^\sharp &= \left\{
      (u, rs^{-1}, v) \in \Forget(\mathrm{Conv})
      \middle|
      \begin{aligned}
        \operatorname{rwr}^*(C)(u) &= (r, u')
        \\
        \operatorname{rwr}^*(C)(v) &= (s, v')
      \end{aligned}
      \land u' = v'
    \right\}
  \end{align*}
\end{definition}

\begin{theorem}
  \label{thm:rewrite}
  Exhaustive rewriting solves the well-defined conversion problem; namely,
  let $C$ be a well-defining conversion, then:
  \begin{align*}
    C^* &= C^\sharp
  \end{align*}
\end{theorem}

\begin{shortproof}%[Theorem~\ref{thm:rewrite}, Outline]
  By mutual inclusion.
  \begin{enumerate}
  \item Define a relation
    $(\rightsquigarrow_{C^*}) \subseteq \Unit \times \Unit\subeval$:
    \begin{equation*}
      \label{eq:proof-rewrite-rel}
      u \rightsquigarrow_{C^*} (r, u') \iff \conv[C^*]ur{\operatorname{unroot}(u')}
    \end{equation*}
    Show by induction 
    \begin{align*}
      u &\rightsquigarrow_{C^*} \operatorname{eval}(u)
      \\
      u \rightsquigarrow_{C^*} v \implies u &\rightsquigarrow_{C^*} \operatorname{rwr}\subeval(C)(v)
    \end{align*}
    that it contains the graph of the exhaustive rewriting map:
    \begin{equation*}
      \label{eq:proof-rewrite-cont}
      u \rightsquigarrow_{C^*} \operatorname{rwr}^*(C)(u)
    \end{equation*}
    Now assume $(u, rs^{-1}, v) \in C^\sharp$, i.e.:
    \begin{align*}
      \begin{aligned}
        \operatorname{rwr}^*(C)(u) &= (r, u')
        \\
        \operatorname{rwr}^*(C)(v) &= (s, v')
      \end{aligned}
      &&
      u' &= v'
    \end{align*}
    Combining this with the preceding property, both directly for $u$ and
    reversed via \eqref{eq:close-cat-inv} for $v$, we get:
    \begin{align*}
      \conv[C^*]{u}{r}{\operatorname{unroot}(u')} = \conv[C^*]{\operatorname{unroot}(v')}{s^{-1}}{v}
    \end{align*}
    Transitivity via \eqref{eq:close-cat-comp} concludes
    \begin{align*}
      \conv[C^*]{u}{rs^{-1}}{v}
    \end{align*}
    such that $C^\sharp \subseteq C^*$.
  \item For the converse, show that $C^\sharp$ is a closure of $C$, by virtue of
    obeying the axioms,
    \begin{align*}
      \conv[C]urv &\implies \conv[C^\sharp]urv
      \\
      \conv[C^\sharp]urv \land \conv[C^\sharp]{u'}{r'}{v'} &\implies \conv[C^\sharp]{uu'}{rr'\kern-1pt}{vv'}
      \\
      \conv[C^\sharp]urv &\implies \conv[C^\sharp]{u^{-1}}{r^{-1}\kern-1pt}{v^{-1}}
      \\
      &\hspace{-2cm} \conv[C^\sharp]u{\operatorname{pval}(u)}{\conv[C^\sharp]{\operatorname{strip}(u)}{\operatorname{pval}(u)^{-1}}{u}}
    \end{align*}
    which requires rewriting to be exhaustive.  Among such relations, $C^*$ is
    minimal.  %\qed
  \end{enumerate}
\end{shortproof}

This theorem ensures that the following decision algorithm is correct.
  
\begin{definition}[Conversion Algorithm]
  \begin{enumerate}
  \item Given two units $u$ and $v$, rewrite them independently and
    exhaustively to evaluated units $(r, u')$ and $(s, v')$, respectively.
  \item Compare the root units $u'$ and $v'$:
    \begin{enumerate}
    \item If they are identical, then the conversion factor is $rs^{-1}$.
    \item Otherwise, there is no conversion factor.
    \end{enumerate}
  \end{enumerate}
\end{definition}

\begin{example}[Mars Climate Orbiter]
  The subsystems of the Orbiter attempted to communicate using different units
  of \emph{linear momentum}, namely \emph{pound-force-seconds}
  $u = \de{\mathrm{lbf}}\de{\mathrm{s}}$ vs.\ \emph{newton-seconds}
  $v = \de{\mathrm{N}}\de{\mathrm{s}}$.  The relevant base units are
  well-defined w.r.t\ the irreducible SI units as
  \begin{math}
    \conv{\de{\mathrm{lbf}}}1{\de{\mathrm{lb}} \de{g_{\mathrm{n}}}}
  \end{math}
  and
  \begin{math}
    \conv{\de{\mathrm{N}}}1{\delta(\mathrm{kg}) \de{\mathrm{m}} \de{\mathrm{s}}^{-2}}
  \end{math},
  with the auxiliary units \emph{pound}
  \begin{math}
    \conv{\de{\mathrm{lb}}}{a}{\de{\mathrm{g}}}
  \end{math} (see Example~\ref{ex:uk})
  and \emph{norm gravity}
  \begin{math}
    \conv{\de{g_{\mathrm{n}}}}{b}{\de{\mathrm{m}} \de{\mathrm{s}}^{-2}}
  \end{math}, and the conversion factors $a=453.59237$ and $b=9.80665$.
  Exhaustive rewriting yields
  \begin{math}
    \operatorname{rwr}^*(C)(u) = (a b, \de{\mathrm{g}} \de{\mathrm{m}} \de{\mathrm{s}}^{-1})
  \end{math} and
  \begin{math}
    \operatorname{rwr}^*(C)(v) = \bigl(1000, \de{\mathrm{g}} \de{\mathrm{m}} \de{\mathrm{s}}^{-1}\bigr)
  \end{math},
  and hence
  \begin{math}
    \conv{u}{ab\mathbin/1000}{v}
  \end{math}. 
  This conclusion reconstructs the factor given in the introduction.

  By virtue of Theorems \ref{thm:conv-cat} and \ref{thm:conv-sdc}, the same
  reasoning can be presented in diagrammatic instead of textual form; see Figure~\ref{fig:mars}.
  \begin{figure}
    \begin{equation*}
      \xymatrix{
        &&
        \lfloor \Usym{g} \rfloor
        \\
        &
        \lfloor \Usym{lb} \rfloor
        \ar[ru]^{a}
        \ar[rr]_{a/1000}
        &&
        \delta(\Usym{kg})
        \ar[lu]_{1000}
        \\
        &
        g_{\mathrm{n}}
        \ar[rr]^{b}
        &&
        \lfloor \Usym{m} \rfloor\lfloor \Usym{s} \rfloor^{-2}
        \\
        &
        \lfloor \Usym{lb} \rfloor g_{\mathrm{n}}
        \ar[rr]^{ab/1000}
        &&
        \delta(\Usym{kg}) \lfloor \Usym{m} \rfloor\lfloor \Usym{s} \rfloor^{-2}
        \\
        \lfloor \Usym{lbf} \rfloor
        \ar[ru]^{1}
        \ar[rrrr]_{ab/1000}
        &&&&
        \lfloor \Usym{N} \rfloor
        \ar[lu]_{1}
      }
    \end{equation*}
    \caption{Calculation of the Mars Climate Orbiter Correction Factor}
    \label{fig:mars}
  \end{figure}
\end{example}

\section{Conclusion}

We have presented a group-theoretic formal model of units of measure, to our
knowledge the first one that supports prefixes and arbitrary conversion factors.
The model is based on two composeable monadic structures, and hence has the
compositionality required for reasoning with substitution and for denotational
semantics.  The model is epistemologically stratified, and distinguishes cleanly
between necessary properties and natural operations on the one hand, and
contingent properties and definable interpretations on the other, leading to a
hierarchy of semantic equivalences.

The semantic data structures of the proposed model deviate from historical
practice in small, but important details:
\begin{itemize}
\item Prefixes apply to compound units as a whole, rather than to individual
  base units.
\item Combinations of prefixes are allowed.
\end{itemize}

We have characterized unit conversion rules by a class of ternary relations
equipped with a closure operator, that function both as a group congruence and
as a category, together with a six-tiered hierarchy of subclasses, each
\NEW{satisfying stronger useful algebraic properties} than the preceding.  This
model reconstructs and justifies the reductionistic approach of traditional
scientific unit systems.

\subsection{Related Work}

\cite{Kennedy1994} has found existing related work then to be generally lacking
in both formal rigor and universality.  With regard to the research programme
outlined by \cite{Kennedy1996}, matters have improved only partially; see
\cite{McKeever2020} for a recent critical survey.
Formally rigorous and expressive type systems have been proposed for various
contexts, such as C \cite{HILLS201251}, UML/OCL \cite{Mayerhofer2016}, generic
(but instantiated for Java) \cite{Xiang2020}.  However, in that line of
research, neither flexibly convertible units in general nor prefixes in
particular are supported.

An object-oriented solution with convertible units is described by
\cite{Allen2004}.  Their approach, like Kenne\-dy's, is syntactic; abelian groups
are added as an ad-hoc language extension with ``abelian classes'' and a
normalization procedure.  Prefixes are not recognized.  Unit conversion is
defined always in relation to a fixed reference unit per dimension.  This
implies that all codimensional units are convertible, which is manifestly
unsound; \NEW{see Section~\ref{unit-dim-aargh}}.  Nevertheless, the approach has
been reiterated in later work such as \cite{Cooper2007}.  By contrast,
\cite{Karr1978} had already proposed a more flexible relational approach to unit
conversion, with methods of matrix calculus for checking consistency.

\subsection{Problem Cases Revisited}
\label{cases-rev}
  
\subsubsection{Problem Case: Neutral Elements}

The statement ``The number $1$ is not a dimension'' \NOW{(translated from
  \cite{ISO80000-1-de})} is a rather meaningless outcome of the ambiguity of
informal language.  The \NOW{\emph{expression}} $1$, taken as the multiplicative
notation for the neutral element $\varnothing\subzero$ of the free abelian group
$\Dim = \Abel(\Dim\subbase)$, is of course a perfectly valid dimension; in this
sense the above statement is false.  $1$ is however not a base dimension symbol,
nor is the dimension that it represents the product of any base dimensions; in
that much narrower sense the above statement is true.  The misunderstanding is
manifest in the attribute ``dimensionless'', which in the domain of quantity
values has no meaning other than ``of dimension $1$''.

% Apparently there is a pervasive pattern that the intuition of natural language
% obscures the ``empty'' elements of generated mathematical structures: It is
% quite common in non-technical contexts to say that two sets, such as the loci of
% two geometrical shapes, have ``no intersection''.  Of course there exists an
% intersection in the strict sense of set theory, but it is the set of no points.
% In a similar vein, one might say that the number $1$ has no prime factoring,
% where in actuality it has the factoring into no prime numbers.
% %% 
% We believe that colloquialisms of this kind should have no place in scientific
% literature such as international standards.
  
The contradictory case of ``$1$ is a derived unit'' is more subtle, but no less
\NOW{problematic}.  The expression $1$, taken as the multiplicative notation for
the neutral element $\varnothing\subzero$ of the free abelian group
$\Unit = \Abel(\Unit\subbase)$ is a valid unit.  As above, it is not a product
of any base units.  \NOW{As recognized in \cite{sib9}, it is a necessary
  inhabitant of the group-based model, not subject to explicit definition.
  Since the standard notions of being derived do not distinguish between syntax
  and semantics, they cannot distinguish between \emph{compound} and
  \emph{defined} entities: The former depend syntactically on the base units
  they are made of.  The latter depend only semantically on their compound
  definition.  The unit $1$ is of the former kind, but there is no base unit
  that it depends on; as such it is derived literally from nothing.}

% The concept of being ``derived'', as used in
% \cite{ISO80000-1}, can only be applied meaningfully to base units: It means that
% an otherwise irreducible symbol is defined as equal to an expression over other
% symbols.  But $1$ is not defined in this sense; its meaning is a necessary
% property of the model being group-based, not a contingent property that can be
% asserted by prescriptive publication.

\NOW{Furthermore, although the unit proper $1 = \varnothing\subzero \in \Unit$
  may not be able to take a prefix $p \in \Pref\subbase$, the normalized unit
  $\operatorname{norm}(1) = (\varnothing\subzero, \varnothing\subzero) \in
  \Unit\subnorm$ surely can, the result being
  $(\delta(p), \varnothing\subzero)$.}
  
\subsubsection{Problem Case: Equational Reasoning}

The traditional syntactic structure of prefixed units, in particular the
interacting design decisions that prefixes are simple and bind stronger than
products, is inherently non-compositional; it cannot be fixed to allow the
substitution of complex expressions for atomic symbols in a straightforward and
consistent way.

The model presented above side-steps the problem by considering a pair of
structures: A faithful model of traditional syntax, $\Unit$ is juxtaposed with a
slightly coarser abstract syntax, $\Unit\subnorm$, the two being related
precisely by the natural distributive law
$\operatorname{norm} = \beta_{\Pref,\Unit\subbase}$.  The latter \NEW{model} is a monad,
which entails a number of well-understood virtues regarding the use as a
compositional data structure \cite[e.g.]{DBLP:conf/mpc/Voigtlander08}; monad
unit and multiplication give a natural notion of compositional substitution that
turns reductionistic defining equations into a fully effective algebraic
rewriting system.

\begin{itemize}
\item The SI base unit $\mathrm{kg}$ parses as
  \begin{align*}
    u = \delta_{\Unit\subpre}((\delta_{\Pref\subbase}(\mathrm{k}), \mathrm{g}))
    \in \Unit
  \end{align*}
  and normalizes to
  \begin{align*}
    u\subnorm = \operatorname{norm}(u) = \bigl(\delta_{\Pref\subbase}(\mathrm{k}),
    \delta_{\Unit\subbase}(\mathrm{g})\bigr) \in \Unit\subnorm = \Pair[\Pref](\Unit\subroot)
  \end{align*}
  Attaching a new prefix $\upmu$ is effected non-invasively by forming an intermediate \NEW{two-level} expression
  \begin{align*}
    e = \bigl(\delta_{\Pref\subbase}(\upmu), u\subnorm\bigr) \in \Pair[\Pref]^2(\Unit\subroot)
  \end{align*}
  and flattening it naturally by means of the multiplication of the monad $\Pair[\Pref]$:
  \begin{align*}
    v\subnorm = \mu_{\Pref,\Unit\subnorm}(e) = \bigl(\delta_{\Pref\subbase}(\upmu)\delta_{\Pref\subbase}(\mathrm{k}),
    \delta_{\Unit\subbase}(\mathrm{g})\bigr) \in \Unit\subnorm
  \end{align*}
  The compound prefix \emph{micro-kilo} is distinct from, but has the same
  value as the simple prefix \emph{milli}.
  \begin{align*}
    u\subnorm &\not\simeq\subnorm v\subnorm &
    u\subnorm &\simeq\subeval v\subnorm
  \end{align*}
  For the purpose of correct semantic calculation, it is of no importance
  whether compound prefix notation is allowed or forbidden syntactically; the
  issue arises only when the result of a calculation needs to be printed in a
  specific format.
\item For dealing with centilitres, consider a well-defining conversion that
  just defines the litre as in Example~\ref{ex:precip}:
  \begin{align*}
    \conv[C]{\de{\mathrm{L}}}{1}{\delta((\delta(\mathrm{d}), \mathrm{m}))^3}
  \end{align*}
  Calculation as above would get stuck with the symbolically irreducible
  compound prefix \emph{centi-deci-deci-deci-}.  Conversion, on the hand, can
  work with numerical values.  Since the conversion $C^*$ in question is
  well-defined, it suffices to consider the results of exhaustive rewriting:
  \begin{align*}
    \operatorname{rwr}^*(C)(\delta((\delta(c), L))) &=
    \operatorname{rwr}\subeval(C)(\operatorname{eval}(\delta((\delta(c), L))))
    \\
    &=
    \xi(\Pair[\RatM]\FinSupp(\operatorname{rwr}\subbase(C))(\operatorname{eval}(\delta((\delta(\mathrm{c}), L)))))
    \\
    &=
    \xi(\Pair[\RatM]\FinSupp(\operatorname{rwr}\subbase(C))((\operatorname{pval}(\delta(\mathrm{c})), \delta(\mathrm{L}))))
    \\
    &=
    \xi(\Pair[\RatM]\FinSupp(\operatorname{rwr}\subbase(C))((\operatorname{val}(\mathrm{c}), \delta(\mathrm{L}))))
    \\
    &=
    \xi((\operatorname{val}(\mathrm{c}), \FinSupp(\operatorname{rwr}\subbase(C))(\delta(\mathrm{L}))))
    \\
    &=
    \xi((\operatorname{val}(\mathrm{c}), \delta(\operatorname{rwr}\subbase(C)(\mathrm{L}))))
    \\
    &=
    \xi((\operatorname{val}(\mathrm{c}), \delta(\mu(\Pair[\RatM](\operatorname{eval})(\operatorname{xpd}(C)(\mathrm{L}))))))
    \\
    &=
    \xi((\operatorname{val}(\mathrm{c}), \delta(\mu(\Pair[\RatM](\operatorname{eval})((1, \delta((\delta(\mathrm{d}), \mathrm{m}))^3))))))
    \\
    &=
    \xi((\operatorname{val}(\mathrm{c}), \delta(\mu((1, \operatorname{eval}(\delta((\delta(\mathrm{d}), \mathrm{m}))^3))))))
    \\
    &=
    \xi((\operatorname{val}(\mathrm{c}), \delta(\mu((1, (\operatorname{pval}(\delta(\mathrm{d})^3), \delta(\mathrm{m})^3))))))
    \\
    &=
    \xi((\operatorname{val}(\mathrm{c}), \delta(\mu((1, (\operatorname{val}(\mathrm{d})^3, \delta(\mathrm{m})^3))))))
    \\
    &=
    \xi((\operatorname{val}(\mathrm{c}), \delta((1 \cdot \operatorname{val}(\mathrm{d})^3, \delta(\mathrm{m})^3))))
    \\
    &= \cdots
    \\
    &=
    \mu((\operatorname{val}(\mathrm{c}), (1 \cdot \operatorname{val}(\mathrm{d})^3, \delta(\mathrm{m})^3)))
    \\
    &=
    (\operatorname{val}(\mathrm{c}) \cdot 1 \cdot \operatorname{val}(\mathrm{d})^3, \delta(\mathrm{m})^3)
    \\
    &=
    (10^{-5}, \delta(\mathrm{m})^3)
  \end{align*}
  Thus, the judgement that a centilitre is $10^{-5}$ cubic meters arises as the
  product of the relevant prefix values and the conversion ratio for the litre.
  Note that the natural transformations employed in the calculation serve no
  other purpose but to transport contingent data to the right places.
\item The defining conversion for the \emph{watt} can be seen in
  Table~\ref{tab:si}.  The semantically equivalent notations for the
  \emph{hectowatt} are exactly the coherence (equivalence) class of
  $\delta((\NEW{\delta}(\mathrm{h}), \delta(\mathrm{W})))$.  Among this set, 44 elements
  obey traditional syntax rules, while infinitely many more involve multiple
  base prefixes on the same base unit.
\end{itemize}

\subsubsection{Problem Case: Cancellation}
\label{cases-rev-cancel}
  
Units of measure are entities of calculation.  As such, their usefulness
depends crucially on their algebraic behavior.  Making exceptions in the
algebra of units in order to encode properties of quantities seems like
putting the cart before the horse.  Thus the distinction
\begin{align*}
  \Usym{rad} = \Usym{m}/\Usym{m} \neq 1 \neq
  \Usym{sr} = \Usym{m}^2/\Usym{m}^2
\end{align*}
makes no algebraic sense, yet is the expression of a valid distinction at some
level, which our model should be able to accomodate.

We propose that the adequate level for making the distinction is quantities:
Quantities have both algebraic and \emph{taxonomical} (sub- and supertype)
structure, with the interaction being largely unstudied.  Quantities may be
associated with a default unit to express their values in.

The quantity \emph{plane angle} is defined as the quotient of the quantities
\emph{arc length} and \emph{circle radius}.  Both are subquantities of
\emph{length}, and hence inherit the default unit $\Usym{m}$.  The quantity
\emph{solid angle} is defined as the quotient of the quantities \emph{covered
  surface area} and \emph{sphere radius} squared.  The former inherits the
default unit $\Usym{m}^2$ from its superquantity \emph{area}, the latter
inherits the squared default unit of \emph{length}.

Thus, the equation $\Usym{rad} = \Usym{m}/\Usym{m}$ is an attempt to encode the
genealogy of the quantity \emph{plane angle} into the definition of its default
unit, in terms of the default units of its precursor quantities.  The message is
irreconcilably at odds with an algebraic reading, since the equation
$\Usym{m}/\Usym{m} = 1$ is a necessary property of the group-based model.

In a model with fine-grained conversion relations, the problem is easily
side-stepped: $\Usym{rad}$ as a unit of plane angle does not require any
reductionistic definition at all, let alone one that equates it unintendendly
with other units.  Rather, $\Usym{rad}$ is codimensional with, but can be left
unconvertible to, both $\Usym{sr}$ and $1$, in the well-defined conversion
relation of choice.

In fact, the possibility to have a partition of multiple interconvertible
components within a set of codimensional units behaves like a subtype system.
We leave the intriguing question whether this feature models relevant issues of
quantity ontology, \NEW{in particular whether it is a suitable complement for
  the concept of quantity kinds}, to future research.
  
\subsubsection{Problem Case: Prefix Families}

The group algebra of prefixes can express the semantics of geometric families of
base prefixes in a formally precise manner within the model: Statements of
interest correspond directly to relationships among mathematical objects.  For
example, the statement ``\emph{mega} means \emph{kilo} squared'' is embodied in
the fact that the pair $(\delta(\mathrm{M}), \delta(\mathrm{k})^2) \in \Pref^2$
is contained in the kernel of the model operation $\operatorname{val}$.

\subsubsection{Problem Case: Units and Dimensions}
\label{re-dim-aargh}

The approach that prefers one canonical unit per dimension is quite prevalent in
the foundational literature \cite{Kennedy1996,Allen2004}, motivated by both the
desire to keep things simple and to manage conversion information in a concise
and consistent way.  It is, however, prone to overdoing convertibility.  The
following quotation demonstrates this by jumping to unwarranted conclusions
twice \cite[boldface added]{Allen2004}:
\begin{quotation}
  To convert between measurements in different units of the same dimension, we
  must specify conversion factors between \textbf{various} units of that
  dimension.  A natural place to keep this information is in the definition of a
  unit: each unit specifies how to convert measurements in that unit to
  measurements in \textbf{any other} defined unit (for the same dimension).

  Although the number of such conversion factors is quadratic in the number of
  units, it is not necessary to maintain so many factors explicitly: if we can
  convert between measurements in units $A$ and $B$, and between measurements
  in units $B$ and $C$, then we can convert between measurements in $A$ and
  $C$ via $B$.  \textbf{Thus}, it is sufficient to include in the definition of
  every unit a single conversion factor to a \emph{primary unit} of that
  dimension, and convert between any two \textbf{commensurable} units via their
  common primary unit.
\end{quotation}

In the model presented here, there is no sufficient reason for all
codimensional units to be convertible; the latter equivalence relation can be
seen as a proper refinement of the former.  The placement of conversion data
is described in an object-oriented fashion in \cite{Allen2004}.  Here, by
contrast, we distinguish a (well-)defining conversion and its closure.

The statement ``The sievert is equal to the joule per kilogram [i.e., the
gray]'' from \cite{CGPM1979R5} is likely \NEW{a} category mistake.  It is
obviously false in a literal reading: Three quantity value assignments that
speak about the same irradiation event may assign very different \NEW{numerical
  values} to absorbed dose (in \emph{gray}), equivalent and effective dose (both
in \emph{sievert}), since the relevant medical calculation models crucially
involve context-dependent scaling factors both greater and less than one.

We conjecture that the actually intended message is threefold:
\begin{itemize}
\item The quantities of \emph{absorbed}, \emph{equivalent} and \emph{effective
    dose of ionizing radiation} are specializations of the more general
  \NEW{parent} quantity \emph{dose of ionizing radiation}.  \NEW{They shall be
    regarded as being of different kind.}
\item Absorbed dose inherits the default unit \emph{gray}, reductionistically
  defined as \emph{joule per kilogram}, from its parent quantity; equivalent
  and effective dose do not.
\item The practice to notate the calculation models with unitless
  scaling factors is acknowledged, although a \NEW{coefficient with a}
  unit of \emph{sievert per gray} would improve the formal rigor.
\end{itemize}

\subsection{Future Work}

\subsubsection{Number Theory of Conversion Ratios}

Conversion ratios have been defined axiomatically as positive rational numbers.
The reason is an application of Occam's Razor: The positive rationals are the
simplest multiplicative group, and no need for other numbers is evident from the
examples discussed so far.  However, there are some elephants in the room.

Of course, the conversion between full turns (or rational fractions thereof) and
radians of planar angle famously features the irrational factor $\pi$.
Fortunately, transcendental numbers such as $\pi$ are nicely orthogonal to
rationals, such that $\pi$ itself can be considered a unit of dimension
$\varnothing\subzero$, and treated purely symbolically.

Another application of irrational factors arises in the treatment of logarithms.
Since logarithms to different bases are merely linearly scaled copies of the
same basic function shape, it is possible to consider a \emph{single}, generic
function $\log$ that makes no mention of the base.  The result can then be
expressed in logarithmic units such as \emph{bit}, \emph{nat} or \emph{decibel},
which specify bases $2$, $e$ or $10^{1/10}$, respectively.
  
As usually in number theory, algebraic irrationals are more complicated, because
their powers may overlap with the rationals in complicated ways.  An
illuminating example is provided by the ``DIN'' families of paper sizes
\cite{ISO216}.  The ``A'' family defines a sequence of length units,
$a_0, \dots, a_{11}$, such that paper sizes arise as $A_n = a_n \times a_{n+1}$.
Ideally, perfect self-similarity would be achieved by the sequence being
geometric, with a factor of $a_n/a_{n+1} = \sqrt2$.  The largest paper size in
the family, $A_0$, covers one square meter.  These facts can be specified
unambiguously by a conversion relation
\begin{align*}
  \conv{a_0a_1}{1}{\mathrm{m}^2} &&
  \conv{a_n^2}{2}{a_{n+1}^2}
\end{align*}
but not by a well-defined one.  Consequently, the ensuing conversion rules are
not complete with respect to the algebra of $\sqrt2$.  The standard
\cite{ISO216} uses a different approach to avoiding irrationals: The lengths are
specified as integer multiples of a millimeter, with margins of tolerance such
that the exact geometric solution is a valid ``approximation''.

\subsubsection{Affine Conversion}

Some traditional units of measure have the property that a \NEW{numerical value}
of zero may have different meanings.  The most well-known example are the units
of \emph{temperature}, namely \emph{degrees Celsius} and \emph{Farenheit}, in
relation to the corresponding SI unit \emph{Kelvin}.  Their formal behavior can
be accommodated in the present model, but with some severe limitations.

The first step is to note that in equation~\ref{eq:conv} in
Remark~\ref{rem:conv}, the ratio $r$ is not just a number, but also a linear map
acting from the right on numerical values.  Thus, the middle component of conversion
relations could be generalized to invertible linear maps, and subsequently to
invertible \emph{affine} maps, for example:
\begin{align*}
  \conv{\text{\textdegree}C}{\operatorname\lambda x\mathpunct.\frac95x+32}{\text{\textdegree}F}
\end{align*}

However, the concept of conversion closure fails for this generalization.
Affine maps are not multiplicative: For two affine maps $f$ and $g$, there is in
general no affine map $fg$ such that $fg(xy) = f(x)g(y)$.

Thus the convertibility of affine units does not survive being put into context.
In particular, the product of any two quantity values measured in such units is
already nonsensical.

\subsubsection{Syntax}

Perhaps surprisingly, the syntax of quantity values is, in actual computing
practice, just as controversial and poorly understood as the semantics.

\cite{ISO80000-1,sib9} give some informal notational rules that document the
standard practice within the SI/ISQ, but not an actual grammar, let alone proof
of non-ambiguity.  By contrast, the programming language F\# supports an ad-hoc
syntax that is not compatible with any standard.  The example
\begin{quote}
  \ttfamily 55.0<miles/hour>
\end{quote}
given in \cite{fsharpunits} features angled brackets instead of the standard
invisible multiplication operator, and a plural form, which is explicitly ruled
out by \cite{ISO80000-1,sib9}.

A systematic study of standard-conforming syntax rules for quantity values, and
of their interaction with common rules for arithmetic expressions, would be a
useful basis for the design and implementation of ergonomic and clear unit-aware
programming language extensions.

\begin{ENEW}
  \section*{Acknowledgments}
  Anonymous reviewers have provided useful suggestions for the improvement of
  this article.
\end{ENEW}
\bibliographystyle{fundam}
\bibliography{conv-fi}

\end{document}